\documentclass[journal,10pt,twocolumn,twoside]{IEEEtran}
\usepackage{amsmath,amsfonts,amssymb,amsbsy,bm,paralist,theorem,color}
\usepackage{graphicx,algorithmic,algorithm,relsize}
\usepackage{multicol}
\usepackage{multirow}
\usepackage{subcaption}
\usepackage[caption=false,font=normalsize,labelfont=sf,textfont=sf]{subfig}
\usepackage{cite}
\usepackage{cuted}
\usepackage{pifont}
\usepackage{comment}
\usepackage{xcolor}

\newtheorem{rem}{Remark}

\newtheorem{prop}{Proposition}

\newcommand{\cmark}{\textcolor{green!70!black}{\ding{51}}}  % ✔
\newcommand{\xmark}{\textcolor{red}{\ding{55}}}             % ✘

\graphicspath{{fig/}}

\definecolor{orange}{RGB}{255,107,0}
\definecolor{green}{RGB}{0,160,20}

\begin{document}
\title{Power-Efficient Optimization for Coexisting Semantic and Bit-Based Users in NOMA Networks}

\author{Ximing Xie,~\IEEEmembership{Member,~IEEE}, Fang Fang,~\IEEEmembership{Senior Member,~IEEE}, Lan Zhang,~\IEEEmembership{Member,~IEEE}, \\and Xianbin Wang,~\IEEEmembership{Fellow,~IEEE}
\thanks{Ximing Xie, Fang Fang and Xianbin Wang are with the Department of Electrical and Computer Engineering, Western University, London, ON N6A 3K7, Canada and Fang Fang is also with the Department of Computer Science, Western University, London, ON N6A 3K7, Canada. (e-mail: \{xxie269, fang.fang, xianbin.wang\}@uwo.ca).}
\thanks{Lan Zhang is with the Department of Electrical and Computer Engineering, Clemson University, Clemson, SC 29634 USA. (e-mail: lan7@clemson.edu).}

}\maketitle

%##################################################################
\begin{abstract}
Semantic communications, which focus on transmitting the semantic meaning of data, have been proposed as a novel paradigm for achieving efficient and relevant communication. Meanwhile, non-orthogonal multiple access (NOMA) enhances spectral efficiency by allowing multiple users to share the same spectrum. However, semantic communications are unlikely to fully replace conventional bit-level communications in the near future, as the latter remain dominant. Therefore, integrating semantic users into a NOMA network alongside conventional bit-based users becomes a meaningful approach to improve both transmission and spectrum efficiency. Nonetheless, due to the lack of a mathematical model that accurately characterizes the relationship between the performance of semantic transceivers and wireless resource allocation, enhancing performance through resource optimization remains a challenge. Moreover, successive interference cancellation (SIC), a key technique in NOMA, introduces additional complexity in system design and implementation. To address these challenges, this paper first improves the deep semantic communication (DeepSC) transceiver to make it adaptive to varying wireless transmission conditions. Subsequently, a data-driven regression approach is employed to develop a mathematical model that captures the impact of wireless resources on semantic transceiver performance. In parallel, a multi-cluster hybrid NOMA (H-NOMA) framework is proposed, where each cluster consists of one semantic user and one bit-based user, to mitigate the complexity introduced by SIC. A total transmit power minimization problem is then formulated by jointly optimizing the beamforming design, bandwidth allocation, and semantic symbol factor. The formulated problem is non-convex and challenging to solve directly. To tackle this, a closed-form optimal solution for the beamforming vectors is first derived. Then, a block coordinate descent (BCD)-based algorithm is developed to determine the bandwidth allocation, while an exhaustive search method is used to optimize the semantic symbol factor. Simulation results illustrate the advantages of the semantic communication over the conventional bit-level communication and verify the superior performance of the proposed framework compared with existing benchmark schemes.

\end{abstract}


\section{Introduction}

Rapid growth in connected devices and wireless applications, such as remote healthcare (RHC) \cite{matre20236g}, virtual reality (VR), and augmented reality (AR) \cite{chen2019towards}, is driving an unprecedented increase in data traffic. As multimedia technologies continue to mature, the demand for ubiquitous high-quality communication services has increased, resulting in a significant increase in the volume of data that must be transmitted. This increased data traffic has led to major challenges in wireless communication systems, especially in terms of resource scarcity and spectrum constraints. Addressing these challenges has become crucial to ensuring that future communication systems can meet user expectations. There are two main directions to overcome these challenges, which are improving resource utilization efficiency and reducing overall traffic. One efficient approach to enhance spectrum efficiency is non-orthogonal multiple access (NOMA), which allows multiple users to share the same resource block by allocating different power levels. This approach utilizes superposition coding at the transmitter and successive interference cancellation (SIC) at the receiver to improve spectrum efficiency \cite{ding2017survey, fang2021energy}. On the other hand, semantic communications, which focus on conveying the intended meaning of information instead of transmitting raw data, have garnered considerable attention for reducing the amount of data and improving transmission  efficiency \cite{tong2022nine, yang2022semantic, luo2022semantic}. Recent achievements in deep learning have further empowered semantic communications, which enable the efficient processing of diverse data types, such as text, speech, images, and videos \cite{huang2024flag, weng2023deep, liang2023selection, zhang2023deep}. As a result, it is natural to explore the integration of semantic communications with NOMA networks motivated by these advantages.

\vspace{-0.3cm}
\subsection{Related Works}
Shannon and Weaver first introduced the concept of semantic communications in  1949 \cite{shannon_mathematical_1949}. After that, research on semantic communications has continued to progress steadily, such as the concept of semantic web \cite{lassila2001semantic} and a novel framework of semantic communications \cite{choi2022unified}. With the rapid development in artificial intelligence and machine learning in recent years, semantic communications have entered a new era. Many studies focus on improving performance, particularly in terms of semantic similarity or semantic accuracy, across various data types, including text, speech, images, and videos. The authors in \cite{xie2021deep} proposed a deep semantic communication (DeepSC) transceiver for text transmission, which outperforms conventional schemes. The authors of \cite{huang2024flag} extended DeepSC to a multi-user scenario for text data transmission. Then, the research shifted to various data types. For example, the authors of \cite{weng2023deep} presented a deep learning-enabled semantic communication system that converts speech into text-related semantic features, significantly reducing data requirements while maintaining high performance. The authors of \cite{liang2023selection} proposed an end-to-end semantic communication system for efficient image transmission by implementing a deep learning-based classifier at the sender and a diffusion model at the receiver. The authors of \cite{zhang2023deep} extended the method in \cite{liang2023selection} to transmit videos by converting videos into frames. Besides the above works that primarily focus on enhancing transmission performance, some works investigated semantic communications from a task-oriented perspective. For example, the authors of \cite{xie2021task} established a multi-user semantic communication system called MU-DeepSC, which leverages correlated image and text data for the visual question answering (VQA) task. The authors in \cite{lyu2024semantic} developed a semantic communication system based on deep learning that simultaneously performs image recovery and classification tasks by integrating JSCC. The task of action recognition and semantic segmentation was accomplished by semantic communication based on deep Over-the-Air computation in \cite{wei2025deepair}.\par
The aforementioned studies focus primarily on optimizing semantic communications at the AI level. Specifically, they emphasize designing innovative deep learning based encoders and decoders. However, as most semantic communication devices operate in conjunction with wireless networks, recent research has increasingly focused on wireless resource-aware optimization. In particular, efforts have been made to enhance the performance of semantic communication systems by incorporating advanced wireless techniques, such as NOMA and reconfigurable intelligent surfaces (RIS), as well as by developing efficient resource allocation strategies. For instance, the authors in \cite{li2023non} studied a downlink NOMA scenario where a base station (BS) simultaneously served multiple semantic users. In addition, several studies, including \cite{wang2024privacy,huang2024joint,xie2024infor}, have explored the integration of RIS into semantic communication systems. However, conventional bit-level communications still dominate the wireless communication field at the current stage. First, many Internet of Things (IoT) devices are highly resource-constrained in terms of computation power, memory, and energy consumption, making them incapable of deploying or running advanced AI models required for semantic communication. Second, current 5G networks and their corresponding communication protocols are primarily designed based on bit-level transmission frameworks. As a result, it is unrealistic to completely replace all bit-based devices with semantic devices. Hence, it is more worthwhile to investigate the practical scenario where semantic users and bit-based users co-exist. The authors in \cite{mu2023exploiting} first explored the scenario where one semantic user and one bit-based user coexist. Subsequently, a multi-user scenario where semantic communications and bit-level communications coexist was investigated in \cite{xia2024resource}. However, the system model in \cite{xia2024resource} was limited to a single-input single-output (SISO) configuration. The semantic and bit-level coexisting networks with multiple-input-single-output (MISO) and multiple-input-multiple-output (MIMO) configurations were investigated in \cite{zhang2025beamforming} and  \cite{feng2024harmonizing}, respectively. Nevertheless, NOMA was not considered in any of these works. The concept of employing NOMA to simultaneously serve one semantic user and one bit-based user was proposed in \cite{mu2023exploiting, ji2024toward}, and a semi-NOMA scheme tailored for this setting was further developed in \cite{mu2023semi}.
\begin{table}[t]
\centering
\caption{Comparisons with relevant literature.}
\label{related_work}
\setlength{\tabcolsep}{6pt}
\renewcommand{\arraystretch}{1.2}
\begin{tabular}{|c|c|c|c|c|c|}
\hline
\multirow{2}{*}{Ref.} 
& \multicolumn{3}{c|}{Scenarios} 
& \multicolumn{2}{c|}{Perspectives} \\
\cline{2-6}
& Multi-user 
& Multi-antenna 
& NOMA
& AI 
& Wireless \\
\hline
\cite{mu2023exploiting} & \xmark & \xmark & \cmark & \xmark & \cmark   \\
\hline
\cite{xia2024resource} & \cmark & \xmark & \xmark & \xmark & \cmark   \\
\hline
\cite{zhang2025beamforming} & \cmark & \cmark & \xmark & \cmark & \cmark  \\
\hline
\cite{feng2024harmonizing} & \cmark & \cmark & \xmark & \xmark & \cmark  \\
\hline
\cite{ji2024toward} & \xmark & \xmark & \cmark & \xmark & \cmark  \\
\hline
\cite{mu2023semi} & \xmark & \xmark & \cmark & \xmark & \cmark  \\
\hline
\textbf{Our work} & \cmark & \cmark & \cmark & \cmark & \cmark  \\
\hline
\end{tabular}
\end{table}
\subsection{Motivations and Contributions}
Most of the existing literature on semantic-bit coexisting networks has overlooked the role of AI in semantic communications, focusing solely on the wireless communication perspective. For example, \cite{mu2023exploiting, mu2023semi, xia2024resource, ji2024toward, feng2024harmonizing} treated the semantic transceiver as a pre-trained black box without considering how to design and train a semantic transceiver and only focused on the physical-layer design. To fill this research gap, this paper includes the design and training of a wireless adaptive-DeepSC transceiver. Moreover, existing research on introducing NOMA into semantic-bit coexisting networks is largely limited to the two-user case with a SISO configuration \cite{mu2023exploiting,ji2024toward,mu2023semi}. While these setups offer useful theoretical insights, they fall short of capturing the complexity of real-world communication systems, where a BS typically serves multiple users simultaneously and is equipped with multiple antennas to exploit spatial degrees of freedom. Motivated by this, this paper considers a more challenging scenario in which the BS is equipped with multiple antennas, and multiple semantic users and bit-based users coexist in the network. The comparisons with some relevant literature are presented in Table \ref{related_work}. \par
The main contributions of this paper are summarized as follows:
\begin{itemize}

    \item We enhance the DeepSC semantic transceiver by integrating a CSI-aware module to improve its adaptability under diverse signal-to-noise (SNR) conditions. This design eliminates the need for repeated retraining across varying channels. Furthermore, we develop a mathematical model to characterize the impact of bandwidth and SNR on the performance of the semantic transceiver. This model consists of two components: word rate and BLEU score. The word rate measures how many words can be transmitted per second, reflecting the transmission efficiency of the semantic transceiver. The BLEU score measures the similarity between the transmitted and received sentences, indicating the semantic accuracy of the transceiver. This enables semantic performance to be seamlessly integrated into wireless optimization problems.
    
    \item We propose a multi-cluster hybrid NOMA (H-NOMA) framework in which each cluster consists of one semantic user and one bit-based user, reflecting the coexistence of different communication paradigms in future networks. NOMA is employed within each cluster, whereas orthogonal multiple access (OMA) is used between clusters. A novel two-period transmission protocol is proposed, consisting of a NOMA period and an exclusive period. During the NOMA period, semantic users and bit-based users are served simultaneously. In the exclusive period, only bit-based users are served.
    
    \item We formulate a total transmit power minimization problem that jointly optimizes beamforming, bandwidth allocation, and the semantic symbol factor. Closed-form beamforming solutions are derived for both periods. The bandwidth allocation for the NOMA period is addressed via a block coordinate descent (BCD)-based algorithm, while the bandwidth allocation for the exclusive period is obtained by solving a convex optimization problem. Finally, an exhaustive search is employed to determine the optimal semantic symbol factor.

    \item Simulation results indicate that semantic communications are more robust than bit-level communications under low SNR conditions. In particular, semantic communications can meet target performance by consuming significantly less transmit power. The proposed joint optimization framework requires less power than the benchmark schemes. It also shows good scalability when the number of antennas or clusters increases, making it suitable for large wireless networks in the future.
\end{itemize}

\subsection{Organization and Notation}
The rest of the paper is organized as follows. In Section II, the architecture of the semantic transceiver is introduced, and a mathematical model is developed. Section III presents the system model of the multi-cluster H-NOMA network and formulates a total transmit power minimization problem. In Section IV, a joint optimization framework is proposed. Simulation results are provided in Section V. Finally, Section VI concludes the paper and discusses potential future research directions.\par
\textit{Notations}: $\mathbf{X}$, $\mathbf{x}$, and $x$ represent a matrix, a vector, and a scalar, respectively. $\mathbf{x}^H$ denotes the Hermitian (conjugate transpose) of vector $\mathbf{x}$. $\mathbb{C}^{N \times 1}$ and $\mathbb{C}^{N \times M}$ denote the sets of $N \times 1$ and $N \times M$ complex-valued vectors and matrices, respectively. $||\cdot||$ denotes the Euclidean ($l_2$) norm. $\operatorname{Tr}(\cdot)$ represents the trace of a matrix, and $\operatorname{rank}(\cdot)$ denotes the matrix rank. $\mathbf{I}_N$ is the $N \times N$ identity matrix. The notation $\mathbf{A} \succeq 0$ indicates that matrix $\mathbf{A}$ is positive semidefinite.

\section{Semantic Transceiver and Performance Mathematical Model Design}
In this section, we first present the design of the semantic transceiver investigated in this study, along with the training methodology used. Subsequently, we build a mathematical model to characterize the performance of the semantic transceiver with SNR based on a series of empirical experiments.
\subsection{Semantic Transceiver Design}
 \begin{figure}[h]
     \centering
     \includegraphics[width=0.47\textwidth]{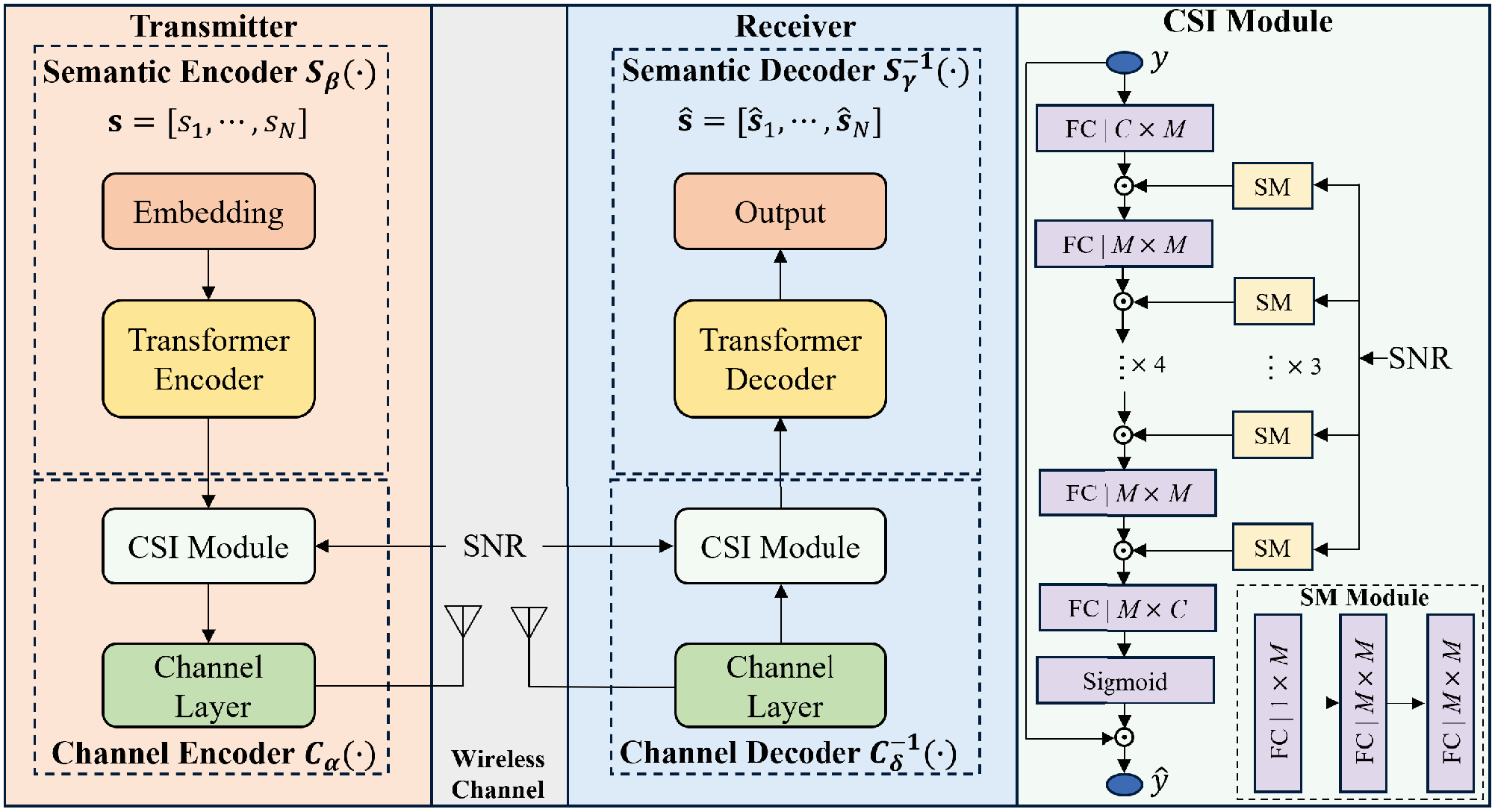}
     \caption{The architecture of semantic transceiver}
     \label{architecture}
\end{figure}
The architecture of the semantic transceiver investigated in this paper is shown in Fig. \ref{architecture}, which is built on the DeepSC proposed in \cite{9398576}. Inspired by \cite{yang2023witt}, we introduce the CSI module into DeepSC to track various channel states and adapt to wireless transmission. In particular, this CSI module makes the semantic transceiver adaptive to different SNR conditions without retraining the entire transceiver. The architecture of CSI module is illustrated in Fig. \ref{architecture}, which consists of 8 fully connected (FC) layers and 3 SNR modules (SM). $C$ denotes the dimension of input and output and $M$ denotes the dimension of hidden embeddings, respectively. SM is implemented as a three-layer FC network, which transforms a one-dimensional scalar SNR to an $M$-dimensional vector. This encoded SNR vector is then applied to the feature vectors through element-wise product in a sequential manner. The activation function ReLU is applied to the output of each FC layer, while a Sigmoid activation function is applied at the final layer. \par

\begin{algorithm}[t]
    \caption{Semantic Transceiver Training and Evaluation}\label{Alg}
    \begin{algorithmic}[1] 
        \STATE {{\bf Training:} Initialize $S_{\beta}(\cdot)$, $S_{\gamma}^{-1}(\cdot)$, $C_{\alpha}(\cdot)$, and $C_{\delta}^{-1}(\cdot)$.}
        \FOR { episode $i = 1,2,...,I$ }
            \FOR{step $t = 1,2,\cdots, T$}
                \STATE {Sample a minibatch (batch size = 64) of data from training dataset.}
                \STATE {Calculate the gradients $\nabla \beta$,$\nabla \gamma$,$\nabla \alpha$, and $\nabla \delta$ through $\mathcal{L}$ and backpropagation.}
                \STATE {Update $S_{\beta}(\cdot)$, $S_{\gamma}^{-1}(\cdot)$, $C_{\alpha}(\cdot)$, and $C_{\delta}^{-1}(\cdot)$ by $\beta^{(t+1)} = \beta^{(t)} - \eta \nabla\beta^{(t)}$, $\gamma^{(t+1)} = \gamma^{(t)} - \eta \nabla\gamma^{(t)}$, $\alpha^{(t+1)} = \alpha^{(t)} - \eta \nabla\alpha^{(t)}$, and $\delta^{(t+1)} = \delta^{(t)} - \eta \nabla\delta^{(t)}$, respectively, with learning rate $\eta$.}
            \ENDFOR
        \ENDFOR
    \end{algorithmic}
     \begin{algorithmic}[1]
        \STATE {{\bf Evaluation:} Freeze $S_{\beta}(\cdot)$, $S_{\gamma}^{-1}(\cdot)$, $C_{\alpha}(\cdot)$, and $C_{\delta}^{-1}(\cdot)$.}
        \FOR { $n = 1,2,...,N$ }
            \STATE {Select a sentence sample denoted by $\mathbf{s}_n$ from test dataset.}
            \STATE {Obtain the recovered sentence $\hat{\mathbf{s}}_n$ through $S_{\beta}(\cdot) \to C_{\alpha}(\cdot) \to$ channel $\to C_{\delta}^{-1}(\cdot) \to S_{\gamma}^{-1}(\cdot)$.}
            \STATE {Calculate 1-gram BLEU score between $\mathbf{s}_n$ and $\hat{\mathbf{s}}_n$.}
        \ENDFOR
        \STATE {Average 1-gram BLEU score over $N$ sentence samples.}
    \end{algorithmic}
\end{algorithm}
Let $\mathbf{s} = [s_1, s_2, \cdots, s_N]$ denote an $N$-word sentence, where $s_i$ denotes the $i$-th word in this sentence. The semantic encoder $S_{\beta}(\cdot)$ with parameter set $\beta$ consists of two parts: the embedding module and the transformer encoder. The embedding module transforms each word to index embedding and then the transformer encoder extracts the semantic information from all index embeddings through a multi-head self-attention mechanism. The semantic information matrix $\mathbf{X} \in \mathbb{R}^{N \times d_s}$, where $d_s$ denotes the dimension of the index embedding, can be expressed as follows:
\begin{equation}
    \mathbf{X} = S_{\beta}(\mathbf{s}). \label{semantic encoder}
\end{equation}
\begin{figure}[t]
     \centering
     \includegraphics[width=0.47\textwidth]{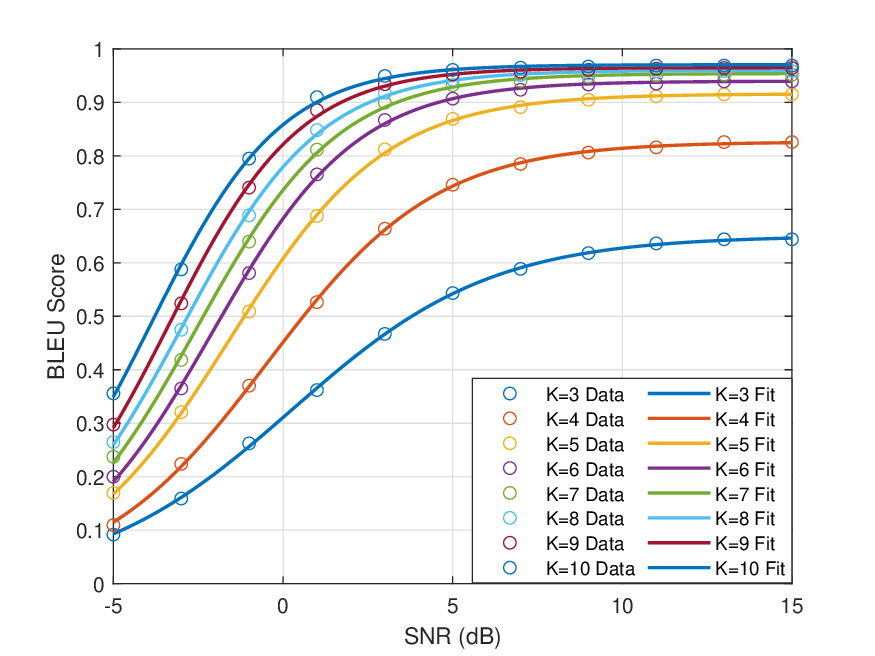}
     \caption{Data regression to fit performance data under different $K$ and SNR.}
     \label{data regression}
\end{figure}
Then, the channel encoder $C_{\alpha}(\cdot)$ with parameter set $\alpha$ further transforms $\mathbf{X}$ into symbols $\mathbf{S} \in \mathbb{C}^{N \times K}$ that can be transmitted over wireless channel as
\begin{equation}
    \mathbf{S} = C_{\alpha} (\mathbf{X}), \label{channel encoder}
\end{equation}
where $K$ is the semantic symbol factor, representing the number of symbols per word. It is worth pointing out that the channel layer in the channel encoder is a FC layer with an output dimension of $2K$. In practice, the value of $K$ can be adjusted to control the number of symbols transmitted over the wireless channel by changing the output dimension of channel layer. The channel decoder $C_{\delta}^{-1}(\cdot)$ and semantic decoder $S_{\gamma}^{-1} (\cdot)$ are the inverse process of $C_{\alpha} (\cdot)$ and $S_{\beta}(\cdot)$, respectively. The recovered sentence $\hat{\mathbf{s}} =  [\hat{s}_1, \hat{s}_2, \cdots, \hat{s}_N] $ can be expressed as
\begin{equation}
    \hat{\mathbf{s}} = S_{\gamma}^{-1} \left(C_{\delta}^{-1}(\hat{\mathbf{S}})\right), \label{recover sentence}
\end{equation}
where $\hat{\mathbf{S}}$ denotes the received symbol matrix from wireless channels. In order to maximize the semantic similarity between the original sentence and the recovered sentence, the loss function is defined as the cross entropy of $\mathbf{s}$ and $\hat{\mathbf{s}}$, which is given by
\begin{equation}
    \mathcal{L} = - \sum\limits_{i=1}^N \left(q(s_i) \log p(\hat{s}_i)\right), \label{loss function of transceiver}
\end{equation}
where $q(s_i)$ is the one-hot representation of $s_i \in \mathbf{s}$, and $p(\hat{s}_i)$ is the predicted probability distribution of the $i$-th word. \par
It is crucial to pre-train the transmitter and receiver before deployment. In particular, we need to train the transceiver to obtain $\beta$, $\alpha$, $\delta$, and $\gamma$. In this paper, we train the semantic transceiver\footnote{The training code is built on the code associated with \cite{wang2023knowledge}.} on WebNLG English dataset\footnote{WebNLG dataset is available at: https://gitlab.com/shimorina/webnlg-dataset/-/tree/master.}. Once the transceiver is trained, we calculate the average 1-gram BLEU score over $N$ sentence samples from the test dataset to evaluate the performance. The training and evaluation process is summarized in Algorithm \ref{Alg}.
\subsection{Semantic Transceiver Performance Mathematical Model}
\begin{table}[t]
\centering
\caption{Logistic fitting parameters}
\begin{tabular}{c|c|c|c}
\hline
$K$ & $A_K$& $l_K$& $x_{0,K}$\\
\hline
3  & 0.650 & 0.340 & 0.262 \\
4  & 0.826 & 0.402 & -0.462 \\
5  & 0.916 & 0.435 & -1.561 \\
6  & 0.940 & 0.469 & -2.084 \\
7  & 0.954 & 0.477 & -2.553 \\
8  & 0.960 & 0.491 & -2.979 \\
9  & 0.965 & 0.516 & -3.379 \\
10 & 0.970 & 0.522 & -3.897 \\
\hline
\end{tabular}\label{parameter}
\end{table}
In semantic communications, symbols convey semantic information rather than bits. Therefore, the Shannon capacity equation used in conventional communications is not suitable as a performance criterion for semantic communications. In order to maximize semantic transceiver performance by optimizing resource allocation over wireless transmission, it is essential to establish a mathematical function that characterizes the relationship between transceiver performance and wireless resources such as bandwidth and transmit power. According to conventional communication transceivers, transmit efficiency and accuracy are two fundamental metrics to assess performance. In this case, we can characterize the performance of a semantic transceiver from these two aspects as well. First, we define word rate to describe transmit efficiency of a text-based semantic transceiver, which can be expressed as
\begin{equation}
    \mathcal{S} = \frac{B}{K}, \label{word rate}
\end{equation}
where $B$ denotes the transmit bandwidth as well as the symbol rate. The word rate measures how many words can be transmitted per second. Then, the transmit accuracy of a text-based semantic transceiver can be measured by 1-gram BLEU score.  However, there is no specific function to bridge BLEU and wireless resources, which is challenging to improve performance by optimizing wireless resource allocation. To formulate a closed-form function of BLEU, the data regression method proposed in \cite{mu2023semi} is utilized to fit performance data. We first train 8 transceiver models with $K$ from $3$ to $10$ and then evaluate each model under different SNR conditions. After that, we use the standard logistic function to fit performance data and the approximated BLEU function is given by
\begin{equation}
    \epsilon_K(x) = \frac{A_K}{1 + e^{-l_K (x - x_{0,K})}}, \label{logistic function}
\end{equation}
where $A_K, l_K$ and $x_{0,K}$ are three parameters related to $K$. $\epsilon_K(x)$ represents the BLEU score of $K$ symbol-based semantic transceiver under SNR condition $x$. Performance data and fitted curves are illustrated in Fig. \ref{data regression} and values of fitting parameters are provided in Table \ref{parameter}. The performance of the proposed semantic transceiver can be mathematically characterized by equations \eqref{word rate} and \eqref{logistic function}.

\section{System Model}
\begin{figure}[t]
     \centering
     \includegraphics[width=0.47\textwidth]{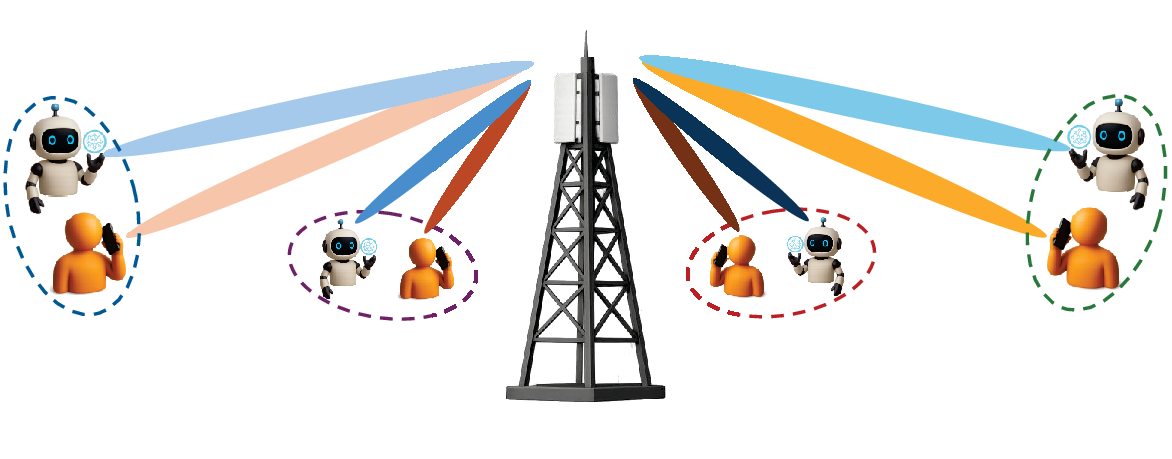}
     \caption{The system model.}
     \label{system}
\end{figure}
The system model is shown in Fig. \ref{system}, which consists of a BS, $M$ semantic users (S-user) and $M$ bit-based users (B-user).  S-users and B-users receive signals in a semantic communication manner and a conventional bit-level communication manner, respectively.  It is assumed that the BS is equipped with $N$ antennas and all semantic and bit-based users are equipped with a single antenna. The BS assigns a dedicated beam to each user. According to \cite{10912507}, H-NOMA can reduce SIC complexity by grouping users into small clusters. Therefore, we allocate users into $M$ clusters and each cluster consists of only one S-user and one B-user. Orthogonal multiple access (OMA) is employed across clusters to eliminate inter-cluster interference by assigning each cluster a distinct bandwidth. Within each cluster, NOMA is adopted, where the S-user and B-user share the same spectrum and apply SIC to remove mutual interference \footnote{Performing SIC between two semantic users is particularly challenging because their semantic transceivers must be jointly pre-trained with specific consideration for interference cancellation. Without such tailored pre-training, one semantic user cannot decode another’s signal. This differs fundamentally from SIC in conventional bit-based systems, where signal structures are standardized. As semantic-level SIC remains an open research problem, the case where two S-users are allocated to the same cluster beyond the scope of this paper and will be investigated in our future work.}. According to \cite{mu2023semi}, the transmitter and receiver for semantic transmission are pre-trained in advance; hence, B-users are impossible to decode the received semantic symbols. In contrast, S-users can employ a separate source and channel coding (SSCC) decoder to decode and recover bit symbols. As a result, a fixed decoding order is applied in each cluster, which is that the S-user removes interference caused by the B-user by SIC before decoding semantic information, while the B-user directly decodes bit information. The intra-cluster transmission protocol is illustrated in Fig. \ref{system model}. The information of the S-user is encoded to semantic symbols by the semantic encoder, while the information of the B-user is encoded to bit symbols by the SSCC encoder. Then, the semantic and bit symbols are superposed to be transmitted through RF chains. Inspired by \cite{zhang2025beamforming}, the entire transmission period is split into three sub-periods, namely pilot period, NOMA period and exclusive period. The pilot period lasts $L_p$ symbol intervals used for channel estimation. The NOMA period lasts $L_{no}$ symbol intervals, during which S-users and B-users are served simultaneously. The exclusive period lasts $L_{ex}$ symbol intervals, during which only B-users are served. In this paper, we focus on the NOMA period and the exclusive period.

 \begin{figure}[t]
     \centering
     \includegraphics[width=0.47\textwidth]{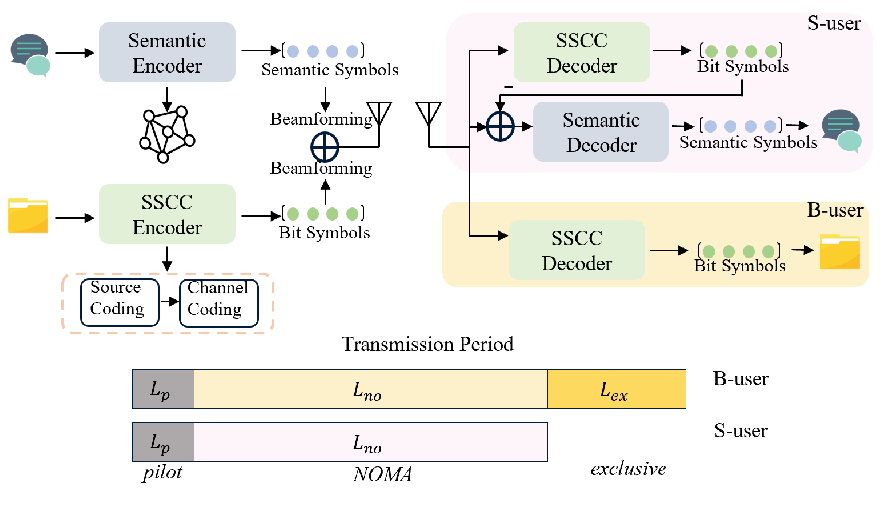}
     \caption{The transmission protocol.}
     \label{system model}
\end{figure}

 % It is known that the complexity of SIC grows significantly with increasing number of users \cite{ding2022state}, therefore, we assume that each cluster consists only of one S-user and one B-user and NOMA is performed only within each cluster. In other words, SIC is conducted only between the S-user and B-user belonging to the same cluster \footnote{The case that each cluster consists of two S-users or two B-users will be investigated in future work. This is because SIC between two S-users requires training new neural network modules and developing new mathematical models, which differ significantly from the traditional SIC process between B-users and fall outside the scope of this paper.}. In addition, according to \cite{9953095}, the transmitter and receiver for semantic transmission are pre-trained in advance, hence, B-users are impossible to decode the received semantic symbols. In contrast, S-users can employ the conventional bit-based encoding and decoding methods to decode and recover bit symbols. As a result, a fixed decoding order is assumed within each cluster: the S-user first decodes the B-user's signal before decoding its own, while the B-user directly decodes its own signal by treating the S-user's signal as interference.

\subsection{NOMA Period}
Let ${\rm U}_{s,i}$ and ${\rm U}_{b,i}$ denote the S-user and B-user in the $i-$th cluster, where $i = \{1,\cdots, M\}$. Since there is no inter-cluster interference, we can consider each cluster independently. Let the $i$-th cluster be an example.  In this period,  ${\rm U}_{s,i}$ and  ${\rm U}_{b,i}$ are served simultaneously. The transmitted signal to the $i-$th cluster at the BS can be expressed as
\begin{equation}
    \mathbf{x}^{no}_i = \mathbf{w}^{no}_{s,i} x_{s,i} + \mathbf{w}^{no}_{b,i} x_{b,i},
\end{equation}
where $\mathbf{w}^{no}_{s,i} \in \mathbb{C}^{N \times 1}$ and $\mathbf{w}^{no}_{b,i} \in \mathbb{C}^{N \times 1}$ denote the beamforming vectors for ${\rm U}_{s,i}$ and  ${\rm U}_{b,i}$, respectively. $x_{s,i}$ and $x_{b,i}$ denote the semantic symbol and bit symbol with unit power, satisfying $\mathbb{E}\left( |x_{s,i}|^2\right) = \mathbb{E}\left( |x_{b,i}|^2\right) = 1$. The superscript $no$ means that the current period is the NOMA period. Then, the received signal at ${\rm U}_{s,i}$ and  ${\rm U}_{b,i}$ can be expressed as
\begin{equation}
    y_{s,i}^{no} = \mathbf{h}_{s,i}^H \mathbf{x}^{no}_i + n_{s,i}, 
\end{equation}
and
\begin{equation}
   y_{b,i}^{no} = \mathbf{h}_{b,i}^H \mathbf{x}^{no}_i + n_{b,i}. 
\end{equation}
$\mathbf{h}_{s,i}$ and $\mathbf{h}_{b,i}$ denote the channel vectors on the BS-${\rm U}_{s,i}$ link and the BS--${\rm U}_{b,i}$ link, respectively. $n_{s,i}$ and $n_{b,i}$ denote additive white Gaussian noise (AWGN) with power spectral density $N_0$. Following the SIC decoding order, the ${\rm U}_{b,i}$'s signal is decoded at $\mathbf{\rm U}_{s,i}$ first. The SINR of ${\rm U}_{b,i}$'s signal decoded at $\mathbf{\rm U}_{s,i}$ can be expressed as
\begin{equation}
    \Gamma_{b \to s, i}^{no} = \frac{|\mathbf{h}_{s,i}^H \mathbf{w}_{b,i}^{no}|^2}{|\mathbf{h}_{s,i}^H \mathbf{w}_{s,i}^{no}|^2 + B_i^{no}N_0}, \label{SINR b to s}
\end{equation}
where $B_i^{no}$ denotes the bandwidth assigned to the $i$-th cluster. After successfully decoding ${\rm U}_{b,i}$'s signal, ${\rm U}_{s,i}$ removes the interference from  ${\rm U}_{b,i}$ and decodes its own signal. The SNR of ${\rm U}_{s,i}$ decoded by itself can be expressed as
\begin{equation}
    \Gamma_{s, i}^{no} = \frac{|\mathbf{h}_{s,i}^H \mathbf{w}_{s,i}^{no}|^2}{B_i^{no}N_0}. \label{SNR s}
\end{equation}
Since ${\rm U}_{b,i}$ directly decodes its own signal with the interference from  ${\rm U}_{s,i}$, the SINR of ${\rm U}_{b,i}$'s signal decoded by itself can be expressed as
\begin{equation}
    \Gamma_{b, i}^{no} = \frac{|\mathbf{h}_{b,i}^H \mathbf{w}_{b,i}^{no}|^2}{|\mathbf{h}_{b,i}^H \mathbf{w}_{s,i}^{no}|^2 + B_i^{no}N_0}. \label{SINR b}
\end{equation}
The achievable data rate of ${\rm U}_{b,i}$ is given by
\begin{equation}
    R_{b,i}^{no} = \min \left(B_i^{no} \log_2 (1 + \Gamma_{b \to s,i}^{no}),  B_i^{no} \log_2 (1 + \Gamma_{b,i}^{no})\right), \label{achievable data rate}
\end{equation}
According to \eqref{word rate} and \eqref{logistic function}, the word rate and transmit accuracy of ${\rm U}_{s,i}$ is $\frac{B_i^{no}}{K}$ and $\epsilon_K \left( \Gamma_{s, i}^{no}\right)$, respectively.

\subsection{Exclusive Period}
In this period, ${\rm U}_{s,i}$ is idle because its transmission has completed. In contrast, ${\rm U}_{b,i}$ continues to receive signals. Since ${\rm U}_{s,i}$ no longer causes interference to ${\rm U}_{b,i}$, the beamforming vector $\mathbf{w}_{b,i}^{no}$ and allocated bandwidth $B_i^{no}$ in the NOMA period may no longer be optimal for ${\rm U}_{b,i}$ in the exclusive period. As a result, we design a new beamforming vector $\mathbf{w}_{b,i}^{ex}$ and re-allocate bandwidth $B_i^{ex}$ for ${\rm U}_{b,i}$ in this period. The superscript $ex$ means that the current period is the exclusive period. The transmitted signal to the $i-$th cluster becomes
\begin{equation}
    \mathbf{x}_i^{ex} = \mathbf{w}^{ex}_{b,i} x_{b,i}. \label{ex signle}
\end{equation}
We assume that the channel experiences large-scale fading, which varies slowly over time. Consequently, the channel is considered to remain constant during the two transmission periods. Then, the received signal at ${\rm U}_{b,i}$ can be expressed as
\begin{equation}
    y_{b,i}^{ex} = \mathbf{h}_{b,i}^H \mathbf{x}_i^{ex} + n_{b,i}. \label{ex received signal}
\end{equation}
The SNR when ${\rm U}_{b,i}$ decodes its signal is given by
\begin{equation}
    \Gamma_{b,i}^{ex} = \frac{|\mathbf{h}_{b,i}^H \mathbf{w}_{b,i}^{ex}|^2}{ B_i^{ex}N_0}. \label{ex SNR b}
\end{equation}
and the data rate of ${\rm U}_{b,i}$ in the exclusive period is 
\begin{equation}
    R_{b,i}^{ex} = B_i^{ex} \log_2 (1 + \Gamma_{b,i}^{ex}). \label{ex b data rate}
\end{equation}

\subsection{Problem Formulation}
To minimize the average transmit power over two periods, the first step is to determine the duration of each period. We assume that the file transmitted by the base station to each user is of equal size, containing $N_w$ words. Therefore, the parameter $L_{no}$ is uniform across all S-users if the same semantic transceiver is deployed. Recall that $K$ means how many symbols to represent a word; hence, we have
\begin{equation}
    L_{no} = N_w K .
\end{equation}
The value of $L_{ex}$ can be calculated by the coding and modulation scheme adopted by B-users\footnote{For example, if the B-user employs ASCII encoding, LDPC coding with a rate of $\frac{1}{2}$, and 64QAM modulation, the average number of characters per word is denoted by $N_c$. The required number of symbols in the exclusive period $L_{ex}$ is calculated as $L_{ex} = \left\lceil \frac{8N_wN_c}{3} \right\rceil$, where $\left\lceil \cdot \right\rceil$ denotes the ceiling (round-up) operation.}. Once the period duration is determined, the average transmit power can be calculated by
\begin{equation}
    p = \frac{L_{no}}{L_{no} + L_{ex}} p^{no} + \frac{L_{ex}}{L_{no} + L_{ex}} p^{ex}. \label{average power}
\end{equation}
$p^{no}$ denotes the total transmit power of the NOMA period, which is given by
\begin{equation}
    p^{no} = \sum\limits_{i=1}^M (||\mathbf{w}_{s,i}^{no}||^2 + ||\mathbf{w}_{b,i}^{no}||^2). \label{NOMA power}
\end{equation}
$p^{ex}$ denotes the total transmit power of the exclusive period, which is given by
\begin{equation}
    p^{ex} = \sum\limits_{i=1}^M ||\mathbf{w}_{b,i}^{ex}||^2. \label{exclusive power}
\end{equation}
Then, the total transmit power minimization problem can be formulated as follows:
\begin{subequations}\label{Prob0} 
\begin{align}
{\rm P_0}: \quad &\min_{\{\mathbf{W},\mathbf{b}, K\}} p \label{P00}\\
\text{s.t.} \quad &\frac{B_i^{no}}{K} \geq \mathcal{S}_0, \quad \forall i \label{P01}\\
\quad\quad &\epsilon_K \left(\Gamma_{s, i}^{no}\right) \geq \varepsilon_0, \quad \forall i \label{P02}\\
\quad\quad &R_{b,i}^{no} \geq R_0, \quad \forall i \label{P03}\\
\quad\quad &R_{b,i}^{ex} \geq R_0, \quad \forall i \label{P04}\\
\quad\quad &\sum\limits_{i=1}^M B_i^{no} \leq B_0 \label{P05}\\
\quad\quad &\sum\limits_{i=1}^M B_i^{ex} \leq B_0 \label{P06}\\
\quad\quad &K \geq 1 \label{P07}
\end{align}
\end{subequations}
where $\mathbf{W}$ denotes a beamforming matrix collecting all beamforming vectors and $\mathbf{b}$ denotes the bandwidth vector collecting all the bandwidths allocated to each cluster. $\mathcal{S}_0$ and $\epsilon_0$ denote the target word rate and the target BLEU score for S-users, respectively, and $R_0$ denote the target data rate for B-users. Constraints \eqref{P01} and \eqref{P02} guarantee transmission efficiency and accuracy for S-users. Constraints \eqref{P03} and \eqref{P04} ensure the quality of service (QoS) for B-users during both the NOMA and exclusive periods. Constraints \eqref{P05} and \eqref{P06} ensure that the total allocated bandwidth does not exceed the available bandwidth $B_0$. Constraint \eqref{P07} ensures that each word is encoded into at least one symbol.

\section{Joint Optimization of Beamforming Design, Bandwidth Allocation and Semantic Symbol Factor Configuration}
In this section, we jointly optimize beamforming, bandwidth allocation, and semantic symbol factor $K$. Since these three optimization variables are interdependent, an efficient approach to solve ${\rm P}_0$ is to optimize one variable at a time while keeping the other two fixed. 

\subsection{Beamforming Design}
When solving beamforming, the bandwidth allocation and $K$ are assumed to be fixed. Therefore, once $B_i^{no}, \forall i$, $B_i^{ex}, \forall i$ and $K$ are given, the beamforming sub-problem can be expressed as 
\begin{subequations}\label{Prob1} 
\begin{align}
{\rm P_1}: \quad &\min_{\{\mathbf{W}\}} \; p \label{P10}\\
\text{s.t.} \quad &\epsilon_K \left(\Gamma_{s,i}^{no}\right) \geq \varepsilon_0, \quad \forall i \label{P11}\\
\quad\quad &R_{b,i}^{no} \geq R_0,\quad \forall i \label{P12}\\
\quad\quad &R_{b,i}^{ex} \geq R_0,\quad \forall i \label{P13}
\end{align}
\end{subequations}
Since $K$ is fixed, we can minimize $p^{no}$ and $p^{ex}$ individually to minimize $p$. Hence, we can further decompose ${\rm P}_1$ into two sub-problems based on two transmission periods.

\subsubsection{NOMA Period}
In this period, we focus only on transmission in the NOMA period. Then, the beamforming optimization problem is given by
\begin{subequations}\label{Prob2} 
\begin{align}
{\rm P_2}: \quad \min_{\{\mathbf{W}^{no}\}} &\sum\limits_{i=1}^M ||\mathbf{w}_{s,i}^{no}||^2 + ||\mathbf{w}_{b,i}^{no}||^2 \label{P20}\\
\text{s.t.} \quad & \eqref{P11}, \eqref{P12}, \notag
\end{align}
\end{subequations}
where $\mathbf{W}^{no}$ denotes the beamforming matrix collecting all beamforming vectors in the NOMA period. Since there is no interference between clusters, ${\rm P_2}$ can be decomposed into $M$ sub-problems based on different clusters. Without loss of generality, we consider the $i-$th cluster as an example. The sub-problem in the $i-$th cluster can be expressed as
\begin{subequations}\label{Prob3} 
\begin{align}
{\rm P_3}: \quad \min_{\{\mathbf{w}_{s,i}^{no},\mathbf{w}_{b,i}^{no}\}} &||\mathbf{w}_{s,i}^{no}||^2 + ||\mathbf{w}_{b,i}^{no}||^2 \label{P30}\\
\text{s.t.} \quad & R_{b,i}^{no} \geq R_0 \label{P31}  \\
\quad\quad &R_{b,i}^{ex} \geq R_0. \label{P32}
\end{align}
\end{subequations}
Let $p_i^{no*} = ||\mathbf{w}_{s,i}^{no*}||^2 + ||\mathbf{w}_{b,i}^{no*}||^2$ denote the optimal value of ${\rm P}_3$. Then, the optimal value of ${\rm P}_2$ can be expressed as
\begin{equation}
    p^{no*} = \sum_{i=1}^M p_i^{no*}. \label{optimal power in NOMA period}
\end{equation}
The next step is to solve ${\rm P_3}$. \par

Note that constraints \eqref{P31} and \eqref{P32} are non-convex, which make ${\rm P}_3$ difficult to be solved. Hence, it is necessary to transform them into convex constraints. To deal with constraint \eqref{P11}, we notice that the first order derivative of \eqref{logistic function} is given by
\begin{equation}
    \frac{d\epsilon_K(x)}{dx} = \frac{A_K l_K e^{-l_K(x - x_{0,K})}}{\left(1 + e^{-l_K(x - x_{0,K})}\right)^2}. \label{first order derivative}
\end{equation}
According to Table \ref{parameter}, $A_K$ and $l_k$ are positive; hence, $\frac{d\epsilon_K(x)}{dx} > 0 $ indicates that $\epsilon_K(x)$ is monotonically increasing with $x$. Therefore, $\exists \; \Gamma_{s,0} \in (0, A_k)$ satisfies
\begin{equation}
    \epsilon_K \left(\Gamma_{s,i}^{no}\right) \geq \epsilon_K \left(\Gamma_{s,0}\right) = \epsilon_0. \label{rewrite1}
\end{equation}
Since $\epsilon_K(x)$ is monotonically increasing with $x$, constraint \eqref{rewrite1} can be equivalently transformed to
\begin{equation}
    \Gamma_{s,i}^{no} \geq  \Gamma_{s,0}. \label{rewrite2}
\end{equation}
$\Gamma_{s,0}$ can be calculated via the inverse function of $\epsilon_K(x)$, which is given by
\begin{equation}
    \epsilon_K^{-1}(y) = x_{0,K} - \frac{1}{l_K}\ln\left(\frac{A_K}{y} - 1\right), y \in (0, A_K). \label{inverse function of logistic function}
\end{equation}
Then, $\Gamma_{s,0} = \epsilon_K^{-1} (\epsilon_0)$. Constraint \eqref{P31} has been recast into an SINR constraint \eqref{rewrite2}. As for constraint \eqref{P32}, it can be directly transformed to two SINR/SNR constraints, which are given by
\begin{equation}
    \Gamma_{b \to s,i}^{no} \geq  \Gamma_{b,0,i}^{no} \label{sinr constraint1}
\end{equation}
and
\begin{equation}
    \Gamma_{b,i}^{no} \geq  \Gamma_{b,0,i}^{no}, \label{sinr constraint2}
\end{equation}
where $\Gamma_{b,0,i}^{no} = 2^{\frac{R_0}{B_i^{no}}}-1$. Although constraints \eqref{P31} and \eqref{P32} have been transformed into SINR/SNR constraints, constraints \eqref{rewrite2} \eqref{sinr constraint1} and \eqref{sinr constraint2} are still non-convex. Note that these constraints all consist of quadratic form related to beamforming vectors. An efficient way to deal with quadratic form is semidefinite relaxation (SDR) \cite{luo2010semidefinite}. \par
In SDR, some auxiliary matrices $\mathbf{W}_{s,i}^{no} = \mathbf{w}_{s,i}^{no} \mathbf{w}_{s,i}^{no H}$, $\mathbf{H}_{s,i} = \mathbf{h}_{s,i}\mathbf{h}_{s,i}^{H}$, $\mathbf{W}_{b,i}^{no} = \mathbf{w}_{b,i}^{no} \mathbf{w}_{b,i}^{no H}$, and $\mathbf{H}_{b,i} = \mathbf{h}_{b,i}\mathbf{h}_{b,i}^{H}$ are introduced to replace all quadratic terms. For example, $||\mathbf{w}_{s,i}^{no}||^2$ is replaced by ${\rm Tr}(\mathbf{W}_{s,i}^{no})$ and $|\mathbf{h}_{s,i}^{H} \mathbf{w}_{s,i}^{no}|^2$ is replaced by ${\rm Tr}(\mathbf{H}_{s,i}\mathbf{W}_{s,i}^{no})$. By applying SDR and some algebraic transformations, constraints \eqref{rewrite2} \eqref{sinr constraint1} and \eqref{sinr constraint2} can be transformed to linear constraints. Hence, ${\rm P}_3$ is recast into
\begin{subequations}\label{Prob4} 
\begin{align}
{\rm P_4}: \quad &\min_{\{\mathbf{W}_{s,i}^{no}, \mathbf{W}_{b,i}^{no}\}}  {\rm Tr}(\mathbf{W}_{s,i}^{no}) + {\rm Tr}(\mathbf{W}_{b,i}^{no})\label{P40}\\
\text{s.t.} \quad &  {\rm Tr}(\mathbf{H}_{s,i}\mathbf{W}_{s,i}^{no}) \geq N_i^{no} \Gamma_{s,0} \label{P41} \\
 \quad &  {\rm Tr}(\mathbf{H}_{s,i}\mathbf{W}_{b,i}^{no}) \geq \Gamma_{b,0,i}^{no} {\rm Tr}(\mathbf{H}_{s,i}\mathbf{W}_{s,i}^{no}) + N_i^{no} \Gamma_{b,0,i}^{no} \label{P42} \\
 \quad &  {\rm Tr}(\mathbf{H}_{b,i}\mathbf{W}_{b,i}^{no}) \geq \Gamma_{b,0,i}^{no} {\rm Tr}(\mathbf{H}_{b,i}\mathbf{W}_{s,i}^{no}) + N_i^{no} \Gamma_{b,0,i}^{no} \label{P43} \\
 \quad & \mathbf{W}_{s,i}^{no} \succeq 0 \label{P44} \\
 \quad & \mathbf{W}_{b,i}^{no} \succeq 0 \label{P45} \\
 \quad & {\rm rank}(\mathbf{W}_{s,i}^{no}) = 1 \label{P46} \\
 \quad & {\rm rank}(\mathbf{W}_{b,i}^{no}) = 1, \label{P47}
\end{align}
\end{subequations}
where $N_i^{no} = B_i^{no} N_0$ denotes the noise power of the $i$-th cluster in the NOMA period. Constraints \eqref{P46} and \eqref{P47} arise from $\mathbf{W}_{s,i}^{no} = \mathbf{w}_{s,i}^{no} \mathbf{w}_{s,i}^{no H}$ and $\mathbf{W}_{b,i}^{no} = \mathbf{w}_{b,i}^{no} \mathbf{w}_{b,i}^{no H}$. With \eqref{P46} and \eqref{P47}, $\mathbf{w}_{s,i}^{no}$ and $\mathbf{w}_{b,i}^{no}$ can be reconstructed from $\mathbf{W}_{s,i}^{no}$ and $\mathbf{W}_{b,i}^{no}$, respectively. However, \eqref{P46} and \eqref{P47} are non-convex constraints. To make ${\rm P}_4$ tractable, rank constraints are usually ignored when solving ${\rm P}_4$. Without rank constraints, ${\rm P}_4$ becomes a convex problem, which is given by
\begin{subequations}\label{Prob5} 
\begin{align}
{\rm P_5}: \quad &\min_{\{\mathbf{W}_{s,i}^{no}, \mathbf{W}_{b,i}^{no}\}}  {\rm Tr}(\mathbf{W}_{s,i}^{no}) + {\rm Tr}(\mathbf{W}_{b,i}^{no})\label{P50}\\
\text{s.t.}  & \qquad \quad  \eqref{P41},\eqref{P42},\eqref{P43},\eqref{P44},\eqref{P45}. \notag
\end{align}
\end{subequations}
Since ${\rm P_5}$ is a convex problem, it can be solved by CVX solvers. Let $\mathbf{W}_{s,i}^{no*}$ and $\mathbf{W}_{b,i}^{no*}$ denote the optimal solution of ${\rm P_5}$. If ${\rm rank}(\mathbf{W}_{s,i}^{no*}) = 1$ and ${\rm rank}(\mathbf{W}_{b,i}^{no*}) = 1$, the optimal solutions $\mathbf{w}_{s,i}^{no*}$ and $\mathbf{w}_{b,i}^{no*}$ of ${\rm P_3}$ can be directly recovered from $\mathbf{W}_{s,i}^{no*}$ and $\mathbf{W}_{b,i}^{no*}$ through matrix decomposition, respectively. Otherwise, $\mathbf{w}_{s,i}^{no*}$ and $\mathbf{w}_{b,i}^{no*}$ have to be approximated via Gaussian randomization.
\begin{prop} \label{prop1}
    The rank of $\mathbf{W}_{s,i}^{no*}$ and $\mathbf{W}_{b,i}^{no*}$ can be guaranteed to be 1.
\end{prop}
\begin{proof}
    According to \cite{luo2010semidefinite}, we have
\begin{equation}
    {\rm rank}(\mathbf{W}_{s,i}^{no*})^2 + {\rm rank}(\mathbf{W}_{b,i}^{no*})^2 \leq m, \label{rank constraint}
\end{equation}
where $m = 3$ denotes the number of linear constraints. Since it is not practically feasible for $\mathbf{W}_{s,i}^{no*} = \mathbf{W}_{b,i}^{no*} = \mathbf{0}$, the rank of $\mathbf{W}_{s,i}^{no*}$ and $ \mathbf{W}_{b,i}^{no*}$ cannot be $0$. As a result, $ {\rm rank}(\mathbf{W}_{s,i}^{no*}) = {\rm rank}(\mathbf{W}_{b,i}^{no*}) = 1$ is the only choice when $m = 3$. The proposition is proved.
\end{proof}
Prompted by Proposition \ref{prop1}, $\mathbf{w}_{s,i}^{no*}$ and $\mathbf{w}_{b,i}^{no*}$ can be obtained by solving the convex problem ${\rm P}_5$ via CVX. However, using CVX to solve ${\rm P}_5$ may still have high computational complexity. Moreover, the implicit relationship between beamforming and bandwidth allocation is difficult to characterize, which potentially introduces additional challenges to the optimization of bandwidth allocation. Therefore, we aim to derive a closed-form expression for the beamforming design to explicitly characterize its relationship with bandwidth allocation. \par
The Lagrangian function of ${\rm P_5}$ can be expressed as
\begin{equation}
    \begin{split}
        \mathcal{L} &=  {\rm Tr}(\mathbf{W}_{s,i}^{no}) + {\rm Tr}(\mathbf{W}_{b,i}^{no}) \\
        &+ \lambda_1 \left(N_i^{no} \Gamma_{s,0} - {\rm Tr}(\mathbf{H}_{s,i}\mathbf{W}_{s,i}^{no})\right) \\
        &+ \lambda_2 \left(\Gamma_{b,0,i}^{no} {\rm Tr}(\mathbf{H}_{s,i}\mathbf{W}_{s,i}^{no}) + N_i^{no}\Gamma_{b,0,i}^{no} - {\rm Tr}(\mathbf{H}_{s,i}\mathbf{W}_{b,i}^{no})\right) \\
        &+ \lambda_3 \left(\Gamma_{b,0,i}^{no} {\rm Tr}(\mathbf{H}_{b,i}\mathbf{W}_{s,i}^{no}) + N_i^{no} \Gamma_{b,0,i}^{no} - {\rm Tr}(\mathbf{H}_{b,i}\mathbf{W}_{b,i}^{no})\right) \\
        &- {\rm Tr} (\mathbf{\Lambda}_1 \mathbf{W}_{s,i}^{no}) - {\rm Tr} (\mathbf{\Lambda}_2 \mathbf{W}_{b,i}^{no}),
    \end{split} \label{Lagrangian}
\end{equation}
where $\lambda_{1}$, $\lambda_{2}$, $\lambda_{3}$, $\mathbf{\Lambda}_1$ and $\mathbf{\Lambda}_2$ are Lagrangian multipliers of inequality constraints. Let $\lambda_{1}^*$, $\lambda_{2}^*$, $\lambda_{3}^*$, $\mathbf{\Lambda}_1^*$ and $\mathbf{\Lambda}_2^*$ denote the optimal Lagrangian multipliers. According to the Karush-Kuhn-Tucker (KKT) conditions, the following inequalities hold, which can be formulated as
\begin{equation}
    \lambda_1^*, \lambda_2^*, \lambda_3^* \geq 0,  \label{greater1}
\end{equation}
\begin{equation}
    \mathbf{\Lambda}_1^* \succeq 0, \mathbf{\Lambda}_2^* \succeq 0. \label{greater2}
\end{equation}
According to the stationarity and complementary slackness, we have
\begin{equation}
    \frac{\partial \mathcal{L}}{\partial \mathbf{W}_{s,i}^{no*}} = \mathbf{I}_N - \lambda_1^* \mathbf{H}_{s,i} + \lambda_2^*\Gamma_{b,0,i}^{no}\mathbf{H}_{s,i} + \lambda_3^*\Gamma_{b,0,i}^{no}\mathbf{H}_{b,i} - \mathbf{\Lambda}_1^* = \mathbf{0}, \label{stationarity 1}
\end{equation}
\begin{equation}
    \frac{\partial \mathcal{L}}{\partial \mathbf{W}_{b,i}^{no*}} = \mathbf{I}_N  - \lambda_2^*\mathbf{H}_{s,i} - \lambda_3^*\mathbf{H}_{b,i} - \mathbf{\Lambda}_2^* = \mathbf{0}, \label{stationarity 2}
\end{equation}
\begin{equation}
     \lambda_1^* \left(N_i^{no} \Gamma_{s,0} - {\rm Tr}(\mathbf{H}_{s,i}\mathbf{W}_{s,i}^{no*})\right) = 0, \label{complementary slackness 1}
\end{equation}
\begin{equation}
    \lambda_2^* \left(\Gamma_{b,0,i}^{no} {\rm Tr}(\mathbf{H}_{s,i}\mathbf{W}_{s,i}^{no^*}) + N_i^{no} \Gamma_{b,0,i}^{no} - {\rm Tr}(\mathbf{H}_{s,i}\mathbf{W}_{b,i}^{no^*})\right) = 0, \label{complementary slackness 2}
\end{equation}
\begin{equation}
    \lambda_3^* \left(\Gamma_{b,0,i}^{no} {\rm Tr}(\mathbf{H}_{b,i}\mathbf{W}_{s,i}^{no*}) + N_i^{no} \Gamma_{b,0,i}^{no} - {\rm Tr}(\mathbf{H}_{b,i}\mathbf{W}_{b,i}^{no*})\right) = 0, \label{complementary slackness 3}
\end{equation}

\begin{equation}
    \mathbf{\Lambda}_1 \mathbf{W}_{s,i}^{no*} = \mathbf{0}, \label{complementary slackness 4}
\end{equation}
and
\begin{equation}
    \mathbf{\Lambda}_2 \mathbf{W}_{b,i}^{no*} = \mathbf{0}.\label{complementary slackness 5}
\end{equation}
$\mathbf{I}_N$ denotes the identity matrix and $\mathbf{0}$ denotes the matrix with all elements are 0. Assisted by Proposition \ref{prop1}, we have $ \mathbf{W}_{s,i}^{no*} =  \mathbf{w}_{s,i}^{no*} \mathbf{w}_{s,i}^{no*H}$ and $ \mathbf{W}_{b,i}^{no*} =  \mathbf{w}_{b,i}^{no*} \mathbf{w}_{b,i}^{no*H}$. After some algebraic transformations, the KKT conditions can be recast into
\begin{equation}
        \lambda_1^* \left(N_i^{no} \Gamma_{s,0} - \mathbf{h}_{s,i}^H \mathbf{w}_{s,i}^{no*} \mathbf{w}_{s,i}^{no*H} \mathbf{h}_{s,i}\right) = 0, \label{KKT1}
    \end{equation}
\begin{align}
     &\lambda_2^* \left(\Gamma_{b,0,i}^{no} \mathbf{h}_{s,i}^H \mathbf{w}_{s,i}^{no*} \mathbf{w}_{s,i}^{no*H} \mathbf{h}_{s,i}\right) \notag \\
     &+ \lambda_2^* \left(N_i^{no} \Gamma_{b,0,i}^{no} - \mathbf{h}_{s,i}^H \mathbf{w}_{b,i}^{no*} \mathbf{w}_{b,i}^{no*H} \mathbf{h}_{s,i}\right) = 0, \label{KKT2}
\end{align}
\begin{align}
     &\lambda_3^* \left(\Gamma_{b,0,i}^{no} \mathbf{h}_{b,i}^H \mathbf{w}_{s,i}^{no*} \mathbf{w}_{s,i}^{no*H} \mathbf{h}_{b,i}\right) \notag \\
     &+ \lambda_3^* \left(N_i^{no} \Gamma_{b,0,i}^{no} - \mathbf{h}_{b,i}^H \mathbf{w}_{b,i}^{no*} \mathbf{w}_{b,i}^{no*H} \mathbf{h}_{b,i}\right) = 0, \label{KKT3}
\end{align}
\begin{align}
    &\mathbf{w}_{s,i}^{no*} - \lambda_1^*\mathbf{h}_{s,i}\mathbf{h}_{s,i}^H\mathbf{w}_{s,i}^{no*} \notag \\
    &+ \lambda_2^*\Gamma_{b,0,i}^{no}\mathbf{h}_{s,i}\mathbf{h}_{s,i}^H\mathbf{w}_{s,i}^{no*} + \lambda_3^*\Gamma_{b,0,i}^{no}\mathbf{h}_{b,i}\mathbf{h}_{b,i}^H\mathbf{w}_{s,i}^{no*} = 0, \label{KKT4}
\end{align}
\begin{equation}
    \mathbf{w}_{b,i}^{no*} - \lambda_2^*\mathbf{h}_{s,i}\mathbf{h}_{s,i}^H\mathbf{w}_{b,i}^{no*} - \lambda_3^*\mathbf{h}_{b,i}\mathbf{h}_{b,i}^H\mathbf{w}_{b,i}^{no*} = 0. \label{KKT5}
\end{equation}
 The definition defined in \cite{chen2016opt} indicates that the channels $\mathbf{h}_{s,i}$ and $\mathbf{h}_{b,i}$ are quasi-degraded if the optimal solutions satisfy $\lambda_2^* = 0$. According to \cite{chen2016opt}, the case with quasi-degraded channels appears with a considerably high
probability. Therefore, we assume the channels of S-users and B-users in each cluster are quasi-degraded. When the channels are quasi-degraded, the optimal closed-form solutions for $\mathbf{w}_s^{no*}$ and $\mathbf{w}_b^{no*}$ can be derived from the above KKT conditions using the method proposed in \cite{8959161}. Due to the space limitation, we directly provide the solutions, which are given by
\begin{equation}
    \mathbf{w}_{s,i}^{no*\rVert} =\frac{(\mathbf{I}_N||\mathbf{h}_{b,i}||^2 +  \Gamma_{b,0,i}^{no} \mathbf{h}_{b,i} \mathbf{h}_{b,i}^H)^{-1} \mathbf{h}_{s,i}}{||(\mathbf{I}_N||\mathbf{h}_{b,i}||^2 +  \Gamma_{b,0,i}^{no} \mathbf{h}_{b,i} \mathbf{h}_{b,i}^H)^{-1} \mathbf{h}_{s,i}||}, \label{optimal beam s}
\end{equation}
\begin{equation}
    \mathbf{w}_{b,i}^{no*\rVert} = \frac{\mathbf{h}_{b,i}}{||\mathbf{h}_{b,i}||}, \label{optimal beam b},
\end{equation}
\begin{equation}
    p_{s,i}^{no*} = \frac{\Gamma_{s,0} N_i^{no}}{|\mathbf{h}_{s,i}^H \mathbf{w}_{s,i}^{no*\rVert}|^2}, \label{optimal power s}
\end{equation}
\begin{equation}
    p_{b,i}^{no*} = \frac{\Gamma_{b,0,i}^{no} \Gamma_{s,0} N_i^{no} |\mathbf{h}_{b,i}^H \mathbf{w}_{s,i}^{no*\rVert}|^2}{|\mathbf{h}_{s,i}^H \mathbf{w}_{s,i}^{no*\rVert}|^2 ||\mathbf{h}_{b,i}||^2} + \frac{\Gamma_{b,0,i}^{no} N_i^{no}}{||\mathbf{h}_{b,i}||^2}, \label{optimal power b}
\end{equation}
where $\mathbf{w}_{s,i}^{no*\parallel}$ and $\mathbf{w}_{b,i}^{no*\parallel}$ denote the normalized beamforming vectors, and $p_{s,i}^{no*}$ and $p_{b,i}^{no*}$ represent the optimal power allocations associated with the two beams. Therefore, the optimal closed-form solutions for $\mathbf{w}_s^{no*}$ and $\mathbf{w}_b^{no*}$ are expressed as
\begin{align}
\mathbf{w}_s^{no*} &= \sqrt{p_{s,i}^{no*}}\, \mathbf{w}_{s,i}^{no*\parallel}, \\
\mathbf{w}_b^{no*} &= \sqrt{p_{b,i}^{no*}}\, \mathbf{w}_{b,i}^{no*\parallel}.
\end{align}
The optimal transmit power of the $i$-th cluster is $p_i^{no*} = p_{s,i}^{no*} + p_{b,i}^{no*}$. Then, the optimal total transmit power in the NOMA period can be calculated by \eqref{optimal power in NOMA period}.
\begin{rem}
    If channels are not quasi-degraded, the optimal closed-form solution of beamforming cannot be achieved. However, a near-optimal closed-form solution of beamforming can be obtained through method proposed in \cite{10912507}.
\end{rem}
\subsubsection{Exclusive Period}
In this period, S-users have finished receiving signals and only B-users receive signals. Similar to the NOMA period, we consider the $i$-th cluster as an example. Hence, the beamforming optimization problem of the $i$-th cluster can be expressed as
\begin{subequations}\label{Prob6} 
\begin{align}
{\rm P_6}: \quad \min_{\{\mathbf{w}_{b,i}^{ex}\}} & ||\mathbf{w}_{b,i}^{ex}||^2 \label{P60}\\
\text{s.t.} \quad & R_{b,i}^{ex} \geq R_0. \label{P61}
\end{align}
\end{subequations}
Let $p_{b,i}^{ex} = ||\mathbf{w}_{b,i}^{ex}||^2$. Constraint $\eqref{P61}$ can be rewritten as
\begin{equation}
    p_{b,i}^{ex} \geq \frac{N_i^{ex} \Gamma_{b,0,i}^{ex}}{|\mathbf{h}_{b,i}^H \mathbf{\Bar{w}}_{b,i}^{ex}|^2},
\end{equation}
where $N_i^{ex} = B_i^{ex}N_0$ denotes the noise power of the $i$-th cluster in the exclusive period, $\mathbf{\Bar{w}}_{b,i}^{ex} = \frac{\mathbf{w}_{b,i}^{ex}}{||\mathbf{w}_{b,i}^{ex}||}$ denotes the normalized beamforming vector and $\Gamma_{b,0,i}^{ex} = 2^{\frac{R_0}{B_i^{ex}}}-1$. Note that if we want $p_{b,i}^{ex}$ minimum, the term $|\mathbf{h}_{b,i}^H \mathbf{\Bar{w}}_{b,i}^{ex}|^2$ should be maximum.  When the direction of  $\mathbf{\Bar{w}}_{b,i}^{ex}$ is aligned with $\mathbf{h}_{b,i}$, $|\mathbf{h}_{b,i}^H \mathbf{\Bar{w}}_{b,i}^{ex}|^2$ can be maximum. Therefore, the optimal $\mathbf{\Bar{w}}_{b,i}^{ex}$ is 
\begin{equation}
    \mathbf{\Bar{w}}_{b,i}^{ex*} = \frac{\mathbf{h}_{b,i}}{||\mathbf{h}_{b,i}||}
\end{equation} and the optimal $ p_{b,i}^{ex}$ is
\begin{equation}
    p_{b,i}^{ex*} = \frac{N_i^{ex} \Gamma_{b,0,i}^{ex}}{||\mathbf{h}_{b,i}||^2}. \label{power exclusive}
\end{equation}
Then, the optimal total transmit power in the exclusive period is given by
\begin{equation}
    p^{ex*} = \sum\limits_{i=1}^M p_{b,i}^{ex*}. \label{optimal power exclusive period}
\end{equation}

\subsection{Bandwidth Allocation}
\begin{algorithm}[t]
    \caption{BCD-based Bandwidth Allocation}\label{Alg BCD}
    \begin{algorithmic}[1] 
        \STATE {{\bf Initialization:} Initialize bandwidth allocation $B_i^{no(0)} = B_0/M, \forall i$}
        \WHILE{ $p^{no* (t)} - p^{no* (t-1)} > \delta_0$}
            \STATE Update $\zeta_i^{(t)}, \forall i$ and $\kappa_i^{(t)}, \forall i$ by substituting $B_i^{no(t-1)}, \forall i$ into~\eqref{fixed block1} and~\eqref{fixed block2}, respectively.
            \STATE Given $\zeta_i^{(t)}, \forall i$ and $\kappa_i^{(t)}, \forall i$, update $B_i^{no(t)}, \forall i$ by minimizing \eqref{shorten equation} under constraints \eqref{P01}, \eqref{P05}.
            \STATE Update $p^{no* (t)}$ by substituting $B_i^{no(t)}, \forall i$, $\zeta_i^{(t)}, \forall i$, and $\kappa_i^{(t)}, \forall i$ into \eqref{shorten equation}.
            \STATE $t = t+1$.
        \ENDWHILE
    \end{algorithmic}
\end{algorithm}
When optimizing bandwidth allocation, we assume that $K$ is fixed. Therefore, the bandwidth allocation problem can be decomposed into two sub-problems for each period.
\subsubsection{NOMA Period} The bandwidth allocation problem in the NOMA period can be formulated as
\begin{subequations}\label{Prob7} 
\begin{align}
{\rm P_7}: \quad &\min_{\{\mathbf{b}^{no}\}} \; p^{no*} (\mathbf{b}^{no}) \label{P70}\\
\text{s.t.} \quad & \eqref{P01}, \eqref{P05}, \notag
\end{align}
\end{subequations}
where $\mathbf{b}^{no}$ denotes the bandwidth allocation vector containing all the bandwidths allocated to each cluster. \eqref{P70} indicates that $ p^{no*}$ is a function of bandwidth allocation. By combining \eqref{optimal power in NOMA period}, \eqref{optimal beam s}, \eqref{optimal power s}, \eqref{optimal power b}, $ p^{no*}$ can be expressed as
\begin{align}
     p^{no*} &= \sum\limits_{i=1}^{M} \frac{\Gamma_{s,0} N_0 B_i^{no}||\mathbf{A}^{-1} \mathbf{h}_{s,i}||^2}{\left|\mathbf{h}_{s,i}^H \mathbf{A}^{-1} \mathbf{h}_{s,i} \right|^2} \notag \\ 
     &+ \sum\limits_{i=1}^{M}\frac{\Gamma_{b,0,i}^{no} N_0 B_i^{no}}{||\mathbf{h}_{b,i}||^2} \left(\frac{\Gamma_{s,0}\left| \mathbf{h}_{b,i}^H \mathbf{A}^{-1} \mathbf{h}_{s,i} \right|^2}{\left| \mathbf{h}_{s,i}^H \mathbf{A}^{-1} \mathbf{h}_{s,i} \right|^2} +1\right), \label{long eqution}
\end{align}
where $\mathbf{A} = \mathbf{I}_N||\mathbf{h}_{b,i}||^2 +  \Gamma_{b,0,i}^{no} \mathbf{h}_{b,i} \mathbf{h}_{b,i}^H$. It is noted that $\Gamma_{b,0,i}^{no}$ and $\mathbf{A}$ are two functions of $B_i^{no}$. Thus, the total transmit power $p^{no*}$ is a non-convex function with respect to $B_i^{no}$ due to the intricate relationship between the bandwidth, achievable rate, SINR constraints, and beamforming design. As a result, directly optimizing $B_i^{no}, \forall i$ is challenging and computationally intractable. To address this challenge, the block coordinate descent (BCD) approach is adopted to iteratively optimize  $B_i^{no}, \forall i$. We first introduce two auxiliary variables to segment \eqref{long eqution} into blocks, which are given by
\begin{equation}
    \zeta_i = \left| \frac{\mathbf{h}_{s,i}^H\mathbf{A}^{-1} \mathbf{h}_{s,i}}{||\mathbf{A}^{-1} \mathbf{h}_{s,i}||} \right|^2 \label{fixed block1}
\end{equation}
and
\begin{equation}
    \kappa_i = \left| \frac{\mathbf{h}_{b,i}^H\mathbf{A}^{-1} \mathbf{h}_{s,i}}{||\mathbf{A}^{-1} \mathbf{h}_{s,i}||} \right|^2. \label{fixed block2}
\end{equation}
By substituting $\Gamma_{b,0,i}^{no}$ with $2^{\frac{R_0}{B_i^{no}}}-1$, \eqref{long eqution} can be recast into
\begin{align}
    p^{no*} &= \sum\limits_{i=1}^M \frac{\Gamma_{s,0} N_0 B_i^{no}}{\zeta_i} \notag \\
    &+ \sum\limits_{i=1}^M\frac{N_0 B_i^{no} }{||\mathbf{h}_{b,i}||^2 } \left(2^{\frac{R_0}{B_i^{no}}}-1\right) \left(\frac{\Gamma_{s,0} \kappa_i}{\zeta_i} + 1\right). \label{shorten equation}
\end{align}
% \begin{strip}
%  % Push content to the bottom
%  \vspace*{-0.8cm}
% \begin{align}
%     \hline \notag \\
%     &p^{no*} = \notag \\
%     &\sum\limits_{i=1}^M \left[\frac{\Gamma_{s,0} N_0 B_i^{no}}{\left| \frac{\mathbf{h}_{s,i}^H(\mathbf{I}_N||\mathbf{h}_{b,i}||^2 +  \left(2^{\frac{R_0}{B_i^{no}}}-1\right) \mathbf{h}_{b,i} \mathbf{h}_{b,i}^H)^{-1} \mathbf{h}_{s,i}}{||(\mathbf{I}_N||\mathbf{h}_{b,i}||^2 +  \left(2^{\frac{R_0}{B_i^{no}}}-1\right) \mathbf{h}_{b,i} \mathbf{h}_{b,i}^H)^{-1} \mathbf{h}_{s,i}||} \right|^2}  + \frac{\left(2^{\frac{R_0}{B_i^{no}}}-1\right) N_0 B_i^{no}}{||\mathbf{h}_{b,i}||^2} \left(\frac{ \Gamma_{s,0}\left| \frac{\mathbf{h}_{b,i}^H(\mathbf{I}_N||\mathbf{h}_{b,i}||^2 +  \left(2^{\frac{R_0}{B_i^{no}}}-1\right) \mathbf{h}_{b,i} \mathbf{h}_{b,i}^H)^{-1} \mathbf{h}_{s,i}}{||(\mathbf{I}_N||\mathbf{h}_{b,i}||^2 +  \left(2^{\frac{R_0}{B_i^{no}}}-1\right) \mathbf{h}_{b,i} \mathbf{h}_{b,i}^H)^{-1} \mathbf{h}_{s,i}||} \right|^2}{\left| \frac{\mathbf{h}_{s,i}^H(\mathbf{I}_N||\mathbf{h}_{b,i}||^2 +  \left(2^{\frac{R_0}{B_i^{no}}}-1\right) \mathbf{h}_{b,i} \mathbf{h}_{b,i}^H)^{-1} \mathbf{h}_{s,i}}{||(\mathbf{I}_N||\mathbf{h}_{b,i}||^2 +  \left(2^{\frac{R_0}{B_i^{no}}}-1\right) \mathbf{h}_{b,i} \mathbf{h}_{b,i}^H)^{-1} \mathbf{h}_{s,i}||} \right|^2}  +1\right) \right]. \label{long eqution}
% \end{align}
% \vspace*{-1cm}
% \end{strip}
$\zeta_i$ and $\kappa_i$ are fixed blocks, which are determined by the value of $B_i^{no}$ from the last iteration when updating $B_i^{no}$. By introducing $\zeta_i$ and $\kappa_i$, the non-convexity caused by their coupling with $B_i^{no}$ can be effectively bypassed. Once $\zeta_i$ and $\kappa_i$ are fixed, $p^{no*}$ is convex with respect to $B_i^{no}$, which can be proved by the second-order derivation
\begin{equation}
    \frac{\partial^2 p^{no*}}{\partial B_i^{no~2}} = \frac{C(R_0 \ln2)^2}{B_i^{no~3}} 2^{\frac{R_0}{B_i^{no}}} > 0, \label{second order}
\end{equation}
where $C = \frac{N_0}{||\mathbf{h}_{b,i}||^2 }\left(\frac{\Gamma_{s,0} \kappa_i}{\zeta_i} + 1\right) > 0$. As a result, the Hessian matrix $\nabla^2 p^{no*}(B_1^{no}, \cdots, B_i^{no})$ is positive definite. With the aid of BCD, problem ${\rm P}_7$ becomes convex in each iteration and can be efficiently solved by CVX. Consequently, the bandwidth allocation in the NOMA period can be optimized accordingly. The proposed BCD-based algorithm is summarized in Algorithm~\ref{Alg BCD}.

\subsubsection{Exclusive Period}
The bandwidth allocation problem in the exclusive period can be formulated as
\begin{subequations}\label{Prob8} 
\begin{align}
{\rm P_8}: \quad &\min_{\{\mathbf{b}^{ex}\}} \; p^{ex*} (\mathbf{b}^{ex}) \label{P80}\\
\text{s.t.} \quad & \eqref{P06}, \notag
\end{align}
\end{subequations}
where $\mathbf{b}^{ex}$ denotes the bandwidth allocation vector containing all the bandwidths allocated to each cluster. According to \eqref{power exclusive}, $p^{ex*}$ can be expressed as
\begin{equation}
    p^{ex*} = \sum\limits_{i=1}^M \frac{N_0 B_i^{ex}}{||\mathbf{h}_{b,i}||^2} \left(2^{\frac{R_0}{B_i^{ex}}}-1 \right).
\end{equation}
Note that ${\rm P}_8$ is a convex problem; hence, the optimal bandwidth allocation in the exclusive period can be obtained by CVX.
\begin{algorithm}[t]
    \caption{Joint Optimization Algorithm}\label{Alg joint}
    \begin{algorithmic}[1]
        \STATE {\bf Initialization:} $p^{opt}$.
        \FOR{$K = 2, 3, \cdots, 9, 10$}
            \STATE Obtain the bandwidth allocation in the NOMA period $\mathbf{b}^{no}$ by Algorithm \ref{Alg BCD}.
            \STATE Obtain the bandwidth allocation in the exclusive period $\mathbf{b}^{ex}$ by solving Problem ${\rm P}_8$.
            \STATE Given $\mathbf{b}^{no}$, calculate $p^{no}$ by \eqref{optimal power in NOMA period}, \eqref{optimal beam s}, \eqref{optimal power s}, \eqref{optimal power b}. 
            \STATE Given $\mathbf{b}^{ex}$, calculate $p^{ex}$ by \eqref{power exclusive}, \eqref{optimal power exclusive period}.
            \STATE Given $p^{no}$ and $p^{ex}$, calculate the average transmit power $p$ by \eqref{average power}.
            \IF{ $p^{opt} > p$}
                \STATE $p^{opt} = p$.
            \ENDIF
        \ENDFOR
        \STATE {\bf Output:} $p^{opt}$.
    \end{algorithmic}
\end{algorithm}

\begin{rem}
    It is worth noting that the proposed beamforming design and bandwidth allocation scheme are also applicable to other semantic transceivers. However, to enable such applicability, a new mathematical model must first be established to characterize the performance of the alternative semantic transceiver with SNR.
\end{rem}

\subsection{Overall Algorithm}
According to Fig. \ref{data regression}, we notice that the number of feasible $K$ is less than $10$. Therefore, the exhaustive search algorithm can be applied to find the optimal $K$. The overall algorithm is summarized in Algorithm \ref{Alg joint}. \par
The computational complexity of Algorithm \ref{Alg joint} is primarily influenced by its optimization structure. Algorithm \ref{Alg joint} exhaustively explores integer values of $K$ from $3$ to $10$. Since the range of $K$ is finite and small, its impact on the overall computational complexity is negligible. Within each iteration for a given $K$, Algorithm \ref{Alg BCD} employs BCD to optimize bandwidth allocation in the NOMA period. Let $T_{\text{BCD}}$ denote the maximum number of BCD iterations. Then, updating auxiliary variables $\zeta_i$ and $\kappa_i$ requires matrix inversions of size $N \times N$, which results in complexity $\mathcal{O}(M N^3)$ per iteration. Additionally, solving the convex optimization problem within Algorithm \ref{Alg BCD} adds complexity $\mathcal{O}(M^3)$ per iteration. Thus, the complexity per BCD iteration is $\mathcal{O}(M N^3 + M^3)$ and the total complexity of Algorithm \ref{Alg BCD} is $\mathcal{O}(T_{\text{BCD}}(M N^3 + M^3))$. Since Problem ${\rm P}_8$ for bandwidth allocation in the exclusive period is solved via CVX, it has the complexity $\mathcal{O}(M^3)$. Combining these two components, the overall computational complexity for Algorithm \ref{Alg joint} is $\mathcal{O}(T_{\text{BCD}}(M N^3 + M^3))$, which is practical and scales polynomially with the number of antennas $N$ and the number of clusters $M$, and linearly with the iteration count $T_{\text{BCD}}$.
\vspace{-0.5cm}
\subsection{Practical Challenges Discussion}
Some practical challenges to implement the proposed H-NOMA framework remain for real-world applications. First, the beamforming design relies on accurate CSI, which may be difficult to obtain or maintain in highly dynamic environments. Second, the semantic transceivers require offline training on representative datasets and must be tailored to the target channel and noise conditions. This raises challenges in adapting the transceiver to unknown or rapidly changing environments without retraining. Finally, the semantic-level SIC between multiple semantic users is still an open research problem, further complicating extensions to more general NOMA scenarios. Addressing these challenges in future work will be crucial for translating theoretical gains into practical systems.
\vspace{-0.3cm}
\section{Numerical Results}
\subsection{Experimental Settings}
In simulations, the channels between the BS and all users are assumed to follow the Rician fading channel model, which is modeled as
\begin{equation}
    \mathbf{h} = \frac{\sqrt{\frac{\kappa_0}{1+\kappa_0}} \mathbf{h}^{\rm LoS} + \sqrt{\frac{1}{1+\kappa_0}} \mathbf{h}^{\rm nLoS}}{\sqrt{d^\mu}},
\end{equation}
where $\mathbf{h}^{\rm LoS}$ is the line-of-sight (LoS) component, $\mathbf{h}^{\rm nLoS}$ is the non-LoS (NLoS) component following the Rayleigh fading model, $\kappa_0$ denotes the Rician factor, $d$ denotes the distance between the BS and the user, and $\mu$ denotes the path loss coefficient. The Rician factor $\kappa_0$ is set as $1$, the distance $d$ is determined by the relative distance between the user and the BS, the path loss coefficient is set as $0.8$. Additionally, the noise power spectral density $N_0$ is set as -140 dBm/Hz, the total transmit bandwidth $B_0$ is set as 1 MHz and the number of clusters $M$ is set as 4.

\subsection{Evaluation of Semantic Communications}
In this subsection, we focus on evaluating the performance of semantic communications. The training process of the proposed semantic transceiver is first illustrated. Then, we compare the performance between semantic communications and conventional bit-level communications. Finally, we evaluate the impact of semantic symbol factor $K$ on the performance. For the conventional bit-level communication, we consider the BS adopts the low density parity check (LDPC) coding scheme and 64QAM to transmit information. In particular, the codeword length is set as $648$, the code rate is set as $\frac{1}{2}$, and the size of the parity check matrix is set as $324 \times 628$. \par
 \begin{figure}[t]
     \centering
     \includegraphics[width=0.47\textwidth]{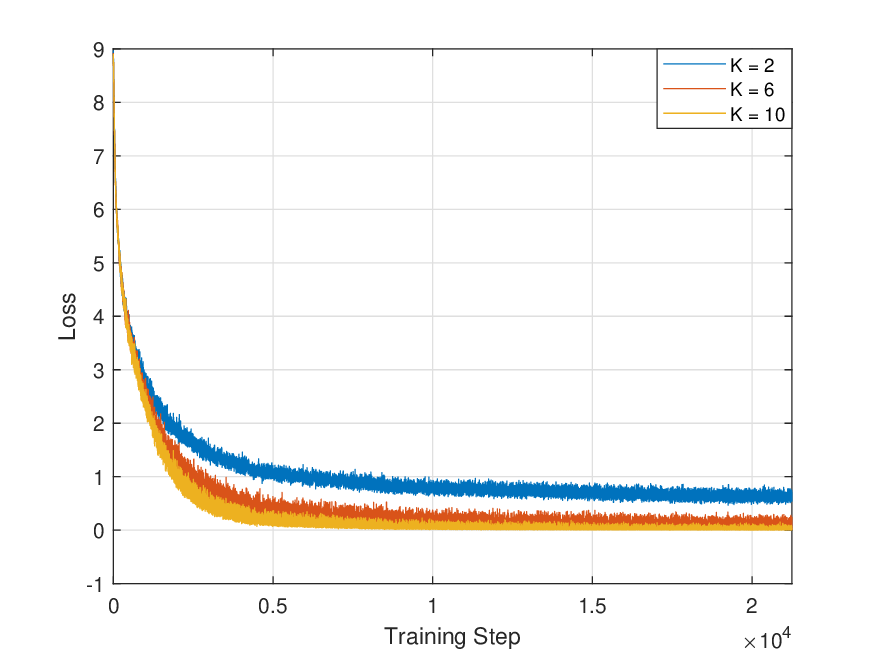}
     \caption{Loss of training the semantic transceiver.}
     \label{loss}
\end{figure}
   Fig.~\ref{loss} shows the training loss curves for different values of the semantic symbol factor $K$, where $K \in \{2, 6, 10\}$. As can be observed, the training loss decreases steadily with the number of training steps and eventually converges. Moreover, the higher the value of $K$, the lower the loss, which indicates that increasing the number of symbols transmitted per word can improve the semantic representation capability of the transceiver. This observation supports the intuition that higher symbol redundancy enables better recovery under noisy channels. \par

 \begin{figure}[t]
     \centering
     \includegraphics[width=0.47\textwidth]{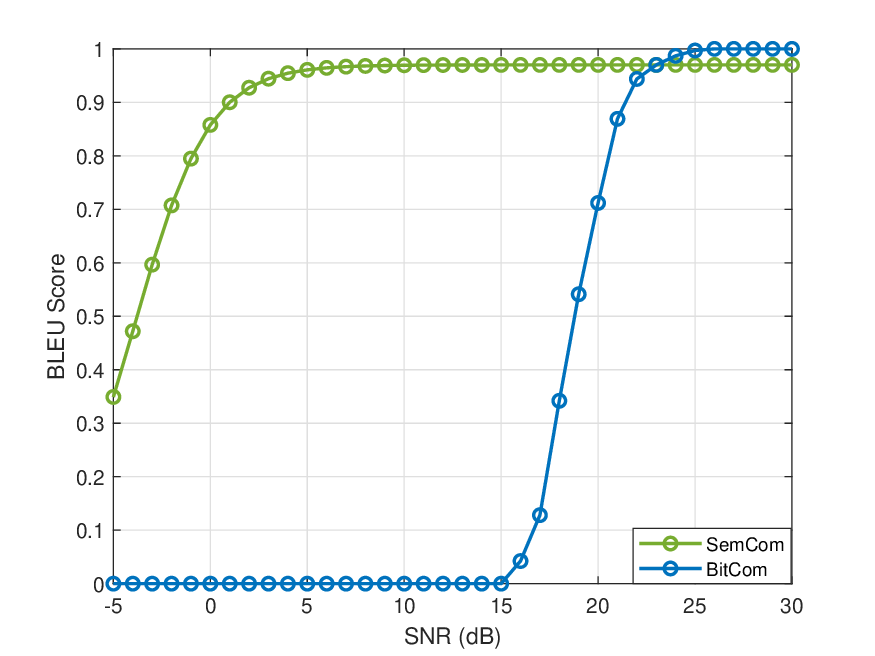}
     \caption{BLEU score comparison between semantic communication and bit-level communication.}
     \label{blue bit semantic}
\end{figure}
Fig.~\ref{blue bit semantic} shows the performance on BLEU scores of the semantic communication and conventional bit-level communication under different SNR conditions. In this simulation, we only consider a single S-user and B-user. As can be observed, the semantic communication scheme achieves significantly higher BLEU scores when SNR is very low. In contrast, the conventional bit-level approach remains ineffective until the SNR surpasses approximately 17 dB, after which it sharply rises to comparable performance levels. This simulation result highlights that semantic communications are more robust than conventional bit-level communications under highly noisy channel conditions, which is more effective for resource-limited wireless environments. Moreover, bit-level communication can achieve perfect accuracy once the SNR exceeds a certain threshold. Therefore, it remains applicable in specific scenarios, such as resource-unconstrained and fully reliable networks. \par

\begin{figure}[t]
     \centering
     \includegraphics[width=0.47\textwidth]{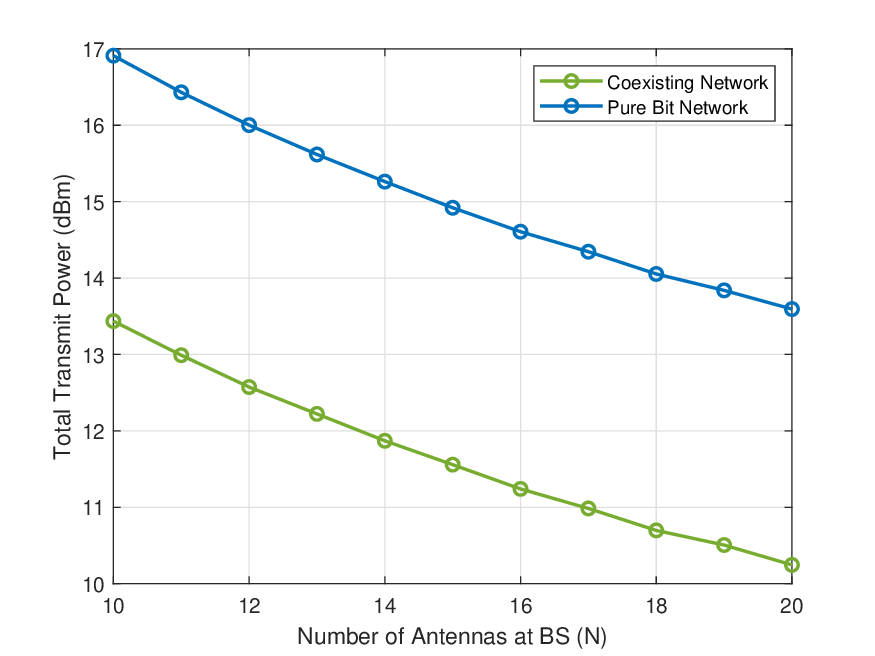}
     \caption{Total transmit power comparison between the proposed coexisting semantic-bit network and the conventional pure bit-level network.}
     \label{power bit semantic}
\end{figure}

\begin{figure}[t]
     \centering
     \includegraphics[width=0.47\textwidth]{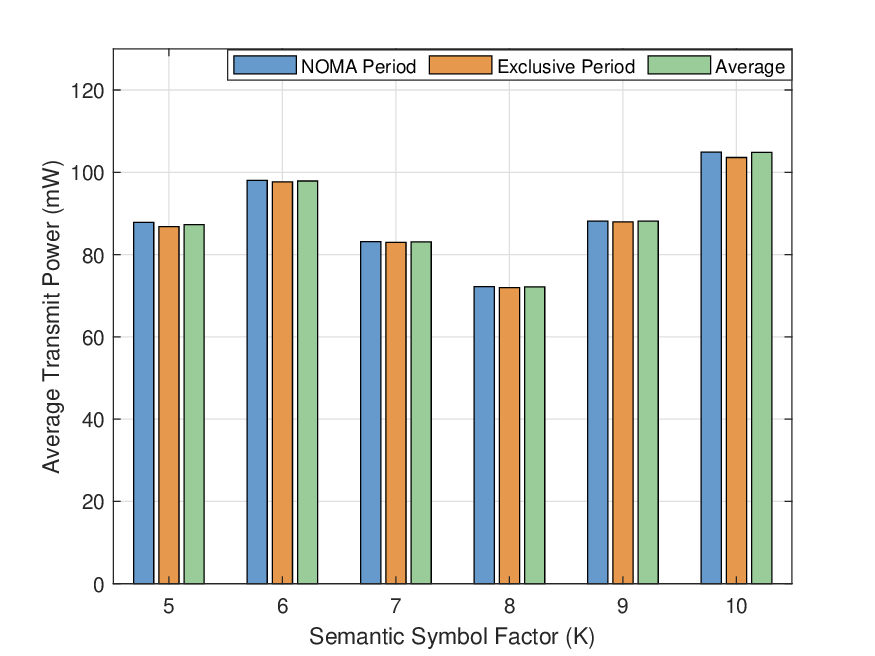}
     \caption{Average transmit power during the NOMA and exclusive periods under different semantic symbol factor $K$.}
     \label{semantic symbol factor}
\end{figure}
Fig.~\ref{power bit semantic} shows the performance on total transmit power of the proposed coexisting semantic-bit NOMA network and the conventional pure bit-level NOMA network. We assume there are 4 clusters. In the coexisting network, each cluster consists of one S-user and B-user. In contrast, each cluster contains two B-users in the conventional pure bit-level network. In this simulation, all users have the same target BLEU score, which is 0.9. As can be observed, the proposed coexisting system consumes less power than the conventional pure bit-level network. The reason can be explained by Fig.~\ref{blue bit semantic}. The required SNR for an S-user to achieve the target BLEU score is below 5 dB, whereas a B-user requires an SNR exceeding 20 dB to reach the same target. It means that the B-user consumes more power to achieve the same performance compared with the S-user. Therefore, the proposed coexisting network has better power efficiency than conventional pure bit-level network. \par
Fig.~\ref{semantic symbol factor} shows the average transmit power of the NOMA period, the exclusive period and the entire transmission period under different semantic symbol factor $K$. As can be observed, the average transmit power in the NOMA period is only slightly higher than that in the exclusive period, which means the network's transmit power is dominated by B-users. The result also shows that the network will consume less transmit power by introducing two transmission periods. It is worth pointing that Fig.~\ref{semantic symbol factor} only shows the result of a one-time experiment. Hence, the optimal $K$ may be varying in different experiments, which depends on CSI. In this specific experiment, the optimal $K$ is 8. An interesting observation is that a larger or smaller value of $K$ does not necessarily lead to better performance. Instead, there exists an optimal intermediate value of $K$ that maximizes system performance.

\subsection{Evaluation of Beamforming Design and Bandwidth Allocation}
In this subsection, we focus on evaluating the beamforming design and bandwidth allocation proposed in this paper. We introduce some beamforming design methods commonly adopted by other literature as benchmarks. The details of benchmarks are summarized as follows:
\begin{itemize}
    \item OMA: In this case, the beamforming is designed based on an orthogonal multiple access (OMA) configuration, without employing SIC. Each user regards the signals from all other users as interference. This beamforming design was adopted by \cite{zhang2025beamforming}. The bandwidth allocation scheme in this case adopts OFDMA, where the total bandwidth is evenly distributed to each cluster. 
    \item ZF: In this case, the beamforming is designed based on zero-forcing. The bandwidth allocation scheme adopts OFDMA.
    \item MRT: In this case, the beamforming is designed based on maximum ratio transmission. The bandwidth allocation scheme adopts OFDMA.
    \item OB-RB: In this case, the beamforming is designed based on the proposed method in this paper whereas the total bandwidth is randomly distributed to each cluster.
\end{itemize}
\begin{figure}[t]
     \centering
     \includegraphics[width=0.47\textwidth]{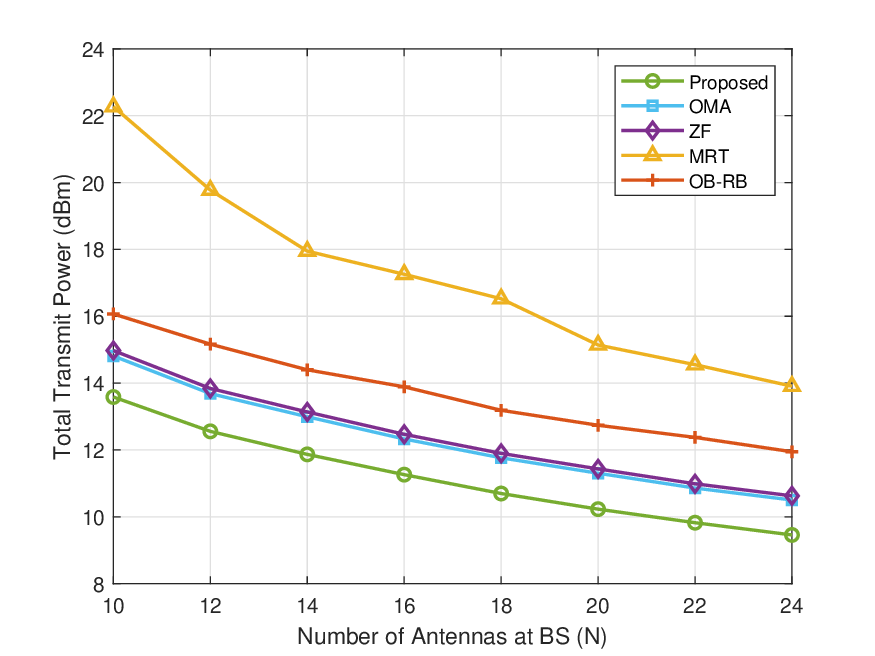}
     \caption{Total transmit power of the proposed scheme and benchmarks under different numbers of BS antennas.}
     \label{antenna}
\end{figure}
Fig.~\ref{antenna} compares the total transmit power of the proposed scheme with various benchmark schemes, including OMA, ZF, MRT, and optimal bandwidth-random beamforming (OB-RB), under different numbers of antennas at the BS. As can be observed, the proposed scheme consistently achieves the lowest total transmit power across the entire range of antenna numbers, which shows the significant advantage of jointly optimizing beamforming, bandwidth allocation, and the semantic symbol factor. Among the benchmarks, MRT exhibits the highest power consumption, while ZF and OMA schemes show moderate improvements. OB-RB performs better than MRT but remains notably less efficient than the proposed method. These results validate the effectiveness and practical benefits of the proposed joint optimization framework, particularly in multi-antenna scenarios.

\begin{figure}[t]
     \centering
     \includegraphics[width=0.47\textwidth]{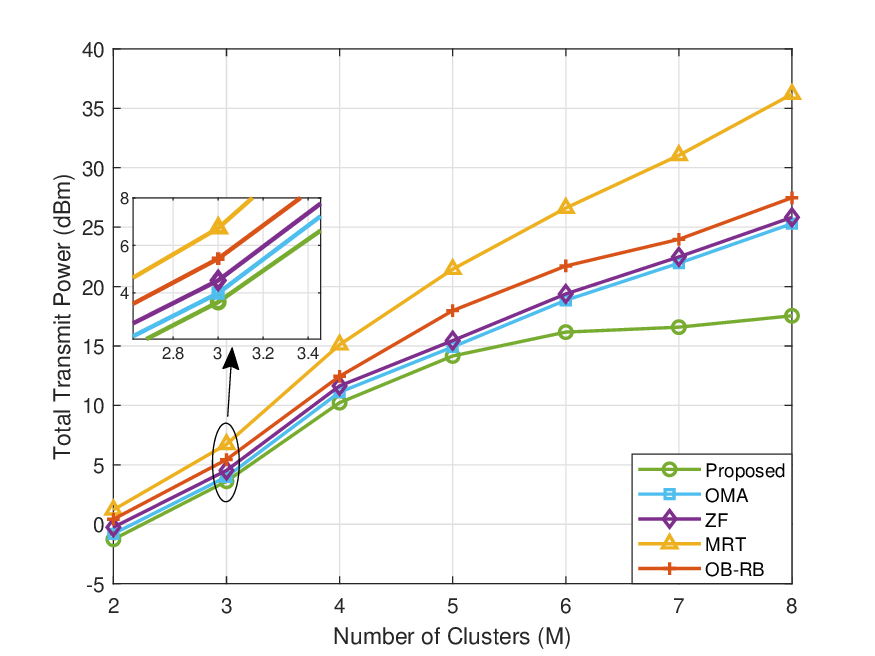}
     \caption{Total transmit power of the proposed scheme and benchmarks under different numbers of clusters.}
     \label{clusters}

\end{figure}

Fig.~\ref{clusters} compares the total transmit power of the proposed scheme with various benchmark schemes, including OMA, ZF, MRT, and OB-RB, under different numbers of clusters. As can be observed, all schemes experience an increase in transmit power with additional clusters, as expected due to the increased number of users and interference management complexity. However, the proposed scheme consistently maintains the lowest transmit power across the entire range of cluster numbers. Specifically, the power gap between the proposed approach and the benchmarks grows larger as more clusters are introduced, which shows that the proposed scheme has high efficiency in scalability. Even with a small number of clusters, the proposed solution still offers improvements over the benchmarks, particularly MRT and OB-RB. These results validate the effectiveness and practical benefits of the proposed joint optimization framework, particularly in multi-cluster scenarios.

\section{Conclusion}
This paper proposed a multi-cluster H-NOMA framework that incorporated multiple semantic and bit-based users. In this framework, the DeepSC semantic transceiver was enhanced with a CSI-aware module to adapt to varying wireless channel conditions. In addition, a transmission protocol was developed to enable the base station to serve both semantic and bit-based users simultaneously using NOMA, where two transmission periods, namely the NOMA period and the exclusive period, were defined. The total transmit power of the proposed system was minimized by jointly optimizing beamforming, bandwidth allocation, and the semantic symbol factor. Simulation results demonstrated that the semantic communication achieved higher robustness at low SNRs and required significantly less power to deliver meaningful content compared to the conventional bit-level communication. Furthermore, results indicated that the proposed joint optimization framework significantly outperformed conventional benchmarks in terms of transmit power efficiency and maintained strong performance as the number of antennas or clusters increased. The findings of this paper indicate that semantic communications serve as a promising foundation for future wireless networks, particularly in power-limited scenarios. Moreover, the proposed semantic–bit coexisting NOMA framework demonstrated good scalability, which is suitable for large-scale deployments.

\bibliographystyle{IEEEtran}
\bibliography{IEEEfull,semantic}

\begin{thebibliography}{10}
\providecommand{\url}[1]{#1}
\csname url@samestyle\endcsname
\providecommand{\newblock}{\relax}
\providecommand{\bibinfo}[2]{#2}
\providecommand{\BIBentrySTDinterwordspacing}{\spaceskip=0pt\relax}
\providecommand{\BIBentryALTinterwordstretchfactor}{4}
\providecommand{\BIBentryALTinterwordspacing}{\spaceskip=\fontdimen2\font plus
\BIBentryALTinterwordstretchfactor\fontdimen3\font minus \fontdimen4\font\relax}
\providecommand{\BIBforeignlanguage}[2]{{%
\expandafter\ifx\csname l@#1\endcsname\relax
\typeout{** WARNING: IEEEtran.bst: No hyphenation pattern has been}%
\typeout{** loaded for the language `#1'. Using the pattern for}%
\typeout{** the default language instead.}%
\else
\language=\csname l@#1\endcsname
\fi
#2}}
\providecommand{\BIBdecl}{\relax}
\BIBdecl

\bibitem{matre20236g}
P.~V. Matre, A.~Kumbhare, R.~Gedam, A.~Sharma, N.~K. Vaishnav, and D.~Naidu, ``6{G} enabled smart iot in healthcare system: Prospect, issues and study areas,'' in \emph{2023 International Conference on Artificial Intelligence for Innovations in Healthcare Industries (ICAIIHI)}, vol.~1.\hskip 1em plus 0.5em minus 0.4em\relax IEEE, 2023, pp. 1--6.

\bibitem{chen2019towards}
M.-H. Chen, K.-W. Hu, I.-H. Chung, and C.-F. Chou, ``Towards {VR/AR} multimedia content multicast over wireless {LAN},'' in \emph{2019 16th IEEE Annual Consumer Communications \& Networking Conference (CCNC)}.\hskip 1em plus 0.5em minus 0.4em\relax IEEE, 2019, pp. 1--6.

\bibitem{ding2017survey}
Z.~Ding, X.~Lei, G.~K. Karagiannidis, R.~Schober, J.~Yuan, and V.~K. Bhargava, ``A survey on non-orthogonal multiple access for 5{G} networks: Research challenges and future trends,'' \emph{IEEE J. Sel. Areas Commun.}, vol.~35, no.~10, pp. 2181--2195, 2017.

\bibitem{fang2021energy}
F.~Fang, K.~Wang, Z.~Ding, and V.~C. Leung, ``Energy-efficient resource allocation for {NOMA-MEC} networks with imperfect {CSI},'' \emph{IEEE Trans. Commun.}, vol.~69, no.~5, pp. 3436--3449, 2021.

\bibitem{tong2022nine}
W.~Tong and G.~Y. Li, ``Nine challenges in artificial intelligence and wireless communications for 6{G},'' \emph{IEEE Wirel. Commun.}, vol.~29, no.~4, pp. 140--145, 2022.

\bibitem{yang2022semantic}
W.~Yang, H.~Du, Z.~Q. Liew, W.~Y.~B. Lim, Z.~Xiong, D.~Niyato, X.~Chi, X.~Shen, and C.~Miao, ``Semantic communications for future internet: Fundamentals, applications, and challenges,'' \emph{IEEE Commun. Surv. Tutor.}, vol.~25, no.~1, pp. 213--250, 2022.

\bibitem{luo2022semantic}
X.~Luo, H.-H. Chen, and Q.~Guo, ``Semantic communications: Overview, open issues, and future research directions,'' \emph{IEEE Wirel. Commun.}, vol.~29, no.~1, pp. 210--219, 2022.

\bibitem{huang2024flag}
W.~Huang, J.~Wang, X.~Chen, Q.~Peng, and Y.~Zhu, ``Flag vector assisted multi-user semantic communications for downlink text transmission,'' \emph{IEEE Commun. Lett.}, vol.~28, no.~6, pp. 1283--1287, 2024.

\bibitem{weng2023deep}
Z.~Weng, Z.~Qin, X.~Tao, C.~Pan, G.~Liu, and G.~Y. Li, ``Deep learning enabled semantic communications with speech recognition and synthesis,'' \emph{IEEE Trans. Wirel. Commun.}, vol.~22, no.~9, pp. 6227--6240, 2023.

\bibitem{liang2023selection}
C.~Liang, D.~Li, Z.~Lin, and H.~Cao, ``Selection-based image generation for semantic communication systems,'' \emph{IEEE Wireless Commun. Lett.}, vol.~28, no.~1, pp. 34--38, 2024.

\bibitem{zhang2023deep}
Z.~Zhang, Q.~Yang, S.~He, and J.~Chen, ``Deep learning enabled semantic communication systems for video transmission,'' in \emph{2023 IEEE 98th Vehicular Technology Conference (VTC2023-Fall)}.\hskip 1em plus 0.5em minus 0.4em\relax IEEE, 2023, pp. 1--5.

\bibitem{shannon_mathematical_1949}
C.~E. Shannon and W.~Weaver, \emph{The mathematical theory of communication}.\hskip 1em plus 0.5em minus 0.4em\relax Champaign, IL: U. Illinois Press, 1949.

\bibitem{lassila2001semantic}
O.~Lassila, J.~Hendler, and T.~Berners-Lee, ``The semantic web,'' \emph{Scientific American}, vol. 284, no.~5, pp. 34--43, 2001.

\bibitem{choi2022unified}
J.~Choi, S.~W. Loke, and J.~Park, ``A unified view on semantic information and communication: A probabilistic logic approach,'' in \emph{2022 IEEE International Conference on Communications Workshops (ICC Workshops)}.\hskip 1em plus 0.5em minus 0.4em\relax IEEE, 2022, pp. 705--710.

\bibitem{xie2021deep}
H.~Xie, Z.~Qin, G.~Y. Li, and B.-H. Juang, ``Deep learning enabled semantic communication systems,'' \emph{IEEE Trans. Signal Process.}, vol.~69, pp. 2663--2675, 2021.

\bibitem{xie2021task}
H.~Xie, Z.~Qin, and G.~Y. Li, ``Task-oriented multi-user semantic communications for {VQA},'' \emph{IEEE Wireless Commun. Lett.}, vol.~11, no.~3, pp. 553--557, 2021.

\bibitem{lyu2024semantic}
Z.~Lyu, G.~Zhu, J.~Xu, B.~Ai, and S.~Cui, ``Semantic communications for image recovery and classification via deep joint source and channel coding,'' \emph{IEEE Trans. Wirel. Commun.}, vol.~23, no.~8, pp. 8388--8404, 2024.

\bibitem{wei2025deepair}
S.~Wei, C.~Feng, C.~Guo, B.~Zhang, and J.~Chen, ``Task-oriented multi-user semantic communication based on deep {Over-the-Air} computation,'' \emph{IEEE Trans. Veh. Technol.}, pp. 1--14, 2025.

\bibitem{li2023non}
W.~Li, H.~Liang, C.~Dong, X.~Xu, P.~Zhang, and K.~Liu, ``Non-orthogonal multiple access enhanced multi-user semantic communication,'' \emph{IEEE Trans. Cogn. Commun. Netw.}, vol.~9, no.~6, pp. 1438--1453, 2023.

\bibitem{wang2024privacy}
Y.~Wang, W.~Yang, P.~Guan, Y.~Zhao, and Z.~Xiong, ``{STAR-RIS}-assisted privacy protection in semantic communication system,'' \emph{IEEE Trans. Veh. Technol.}, vol.~73, no.~9, pp. 13\,915--13\,920, 2024.

\bibitem{huang2024joint}
Y.~Huang, C.~Cai, X.~Yuan, and Y.-J.~A. Zhang, ``Joint active and passive beamforming for {RIS-Aided} semantic communication,'' \emph{IEEE Trans. Veh. Technol.}, vol.~73, no.~12, pp. 19\,815--19\,820, 2024.

\bibitem{xie2024infor}
P.~Xie, F.~Li, M.~Zhang, W.~Quan, J.~Zhu, and N.~Cheng, ``{STAR-RIS} assisted information transmission based on fairness in semantic communication systems,'' \emph{IEEE Trans. Wirel. Commun.}, vol.~23, no.~11, pp. 17\,007--17\,020, 2024.

\bibitem{mu2023exploiting}
X.~Mu and Y.~Liu, ``Exploiting semantic communication for non-orthogonal multiple access,'' \emph{IEEE J. Sel. Areas Commun.}, vol.~41, no.~8, pp. 2563--2576, 2023.

\bibitem{xia2024resource}
L.~Xia, Y.~Sun, D.~Niyato, L.~Zhang, and M.~Ali~Imran, ``Wireless resource optimization in hybrid semantic/bit communication networks,'' \emph{IEEE Trans. Commun.}, vol.~73, no.~5, pp. 3318--3332, 2025.

\bibitem{zhang2025beamforming}
M.~Zhang, G.~Zhu, R.~Jin, X.~Chen, Q.~Shi, C.~Zhong, and K.~Huang, ``Beamforming design for semantic-bit coexisting communication system,'' \emph{IEEE J. Sel. Areas Commun.}, vol.~43, no.~4, pp. 1262--1277, 2025.

\bibitem{feng2024harmonizing}
B.~Feng, X.~Han, and C.~Feng, ``Harmonizing efficiency and precision in semantic-bit coexisting communication systems,'' in \emph{2024 22nd International Symposium on Modeling and Optimization in Mobile, Ad Hoc, and Wireless Networks (WiOpt)}.\hskip 1em plus 0.5em minus 0.4em\relax IEEE, 2024, pp. 111--117.

\bibitem{ji2024toward}
Y.~Ji, B.~Dong, and B.~Cao, ``Toward opportunistic semantic and bit communications: A {NOMA} enabled approach,'' in \emph{2024 IEEE/CIC International Conference on Communications in China (ICCC)}, 2024, pp. 657--662.

\bibitem{mu2023semi}
X.~Mu, Y.~Liu, L.~Guo, and N.~Al-Dhahir, ``Heterogeneous semantic and bit communications: A semi-{NOMA} scheme,'' \emph{IEEE J. Sel. Areas Commun.}, vol.~41, no.~1, pp. 155--169, 2023.

\bibitem{9398576}
H.~Xie, Z.~Qin, G.~Y. Li, and B.-H. Juang, ``Deep learning enabled semantic communication systems,'' \emph{IEEE Trans. Signal Process.}, vol.~69, pp. 2663--2675, 2021.

\bibitem{yang2023witt}
K.~Yang, S.~Wang, J.~Dai, K.~Tan, K.~Niu, and P.~Zhang, ``{WITT}: A wireless image transmission transformer for semantic communications,'' in \emph{ICASSP 2023-2023 IEEE International Conference on Acoustics, Speech and Signal Processing (ICASSP)}.\hskip 1em plus 0.5em minus 0.4em\relax IEEE, 2023, pp. 1--5.

\bibitem{wang2023knowledge}
B.~Wang, R.~Li, J.~Zhu, Z.~Zhao, and H.~Zhang, ``Knowledge enhanced semantic communication receiver,'' \emph{IEEE Communications Letters}, vol.~27, no.~7, pp. 1794--1798, 2023.

\bibitem{10912507}
X.~Xie, F.~Fang, and X.~Wang, ``Rethinking power minimization in a downlink hybrid {NOMA} network,'' \emph{IEEE Commun. Lett.}, vol.~29, no.~5, pp. 953--957, 2025.

\bibitem{luo2010semidefinite}
Z.~Luo, W.~Ma, A.~M.~C. So, Y.~Ye, and S.~Zhang, ``Semidefinite relaxation of quadratic optimization problems,'' \emph{IEEE Signal Process. Mag.}, vol.~27, no.~3, pp. 20--34, 2010.

\bibitem{chen2016opt}
Z.~Chen, Z.~Ding, P.~Xu, and X.~Dai, ``Optimal precoding for a {QoS} optimization problem in two-user {MISO-NOMA} downlink,'' \emph{IEEE Wireless Commun. Lett.}, vol.~20, no.~6, pp. 1263--1266, 2016.

\bibitem{8959161}
J.~Zhu, J.~Wang, Y.~Huang, K.~Navaie, Z.~Ding, and L.~Yang, ``On optimal beamforming design for downlink {MISO} {NOMA} systems,'' \emph{IEEE Trans. Veh. Technol.}, vol.~69, no.~3, pp. 3008--3020, 2020.

\end{thebibliography}

\end{document}